\documentclass[a4paper,twoside,notitlepage]{article}
\usepackage{a4}
\usepackage[english]{babel}
\usepackage{color,graphicx}
\usepackage{amsmath,amssymb,amsthm,amsfonts}
\usepackage{mathrsfs}
\usepackage{verbatim}
\usepackage{bussproofs}
\usepackage{wasysym}
\usepackage{soul}
\usepackage{lscape}
\usepackage{cmll}
\usepackage{tikz-cd}
\usepackage{MnSymbol}
\usepackage{fullpage}

\newcommand{\txt}[1]{\quad\textnormal{#1}\quad}

\newcommand{\lbi}[3]{\mathrm{let}\;#1\;\mathrm{be}\;#2\;\mathrm{in}\;#3}
\long\def\symbolfootnote[#1]#2{\begingroup%
\def\thefootnote{\fnsymbol{footnote}}\footnote[#1]{#2}\endgroup}

\newcommand{\ra}[1]{\stackrel{#1}{\longrightarrow}}

\newtheorem{theorem}{Theorem}
\newtheorem*{theorem*}{Theorem}
\newtheorem{lemma}{Lemma}
\newtheorem*{lemma*}{Lemma}

\newtheorem*{claim*}{Claim}

\newtheorem{corollary}{Corollary}

\theoremstyle{definition}
\newtheorem{definition}{Definition}
\newtheorem*{definition*}{Definition}

\newtheorem{remark}{Remark}

\begin{document}
\title{\textbf{Syntax and Semantics of Linear Dependent Types}
\\ \Large\textbf{Technical Report}}
\author{Matthijs V\'ak\'ar}
\date{Oxford, UK, \today}

\maketitle
\begin{abstract}A type theory is presented that combines (intuitionistic) linear types with type dependency, thus properly generalising both intuitionistic dependent type theory and full linear logic. A syntax and complete categorical semantics are developed, the latter in terms of (strict) indexed symmetric monoidal categories with comprehension. Various optional type formers are treated in a modular way. In particular, we will see that the historically much-debated multiplicative quantifiers and identity types arise naturally from categorical considerations. These new multiplicative connectives are further characterised by several identities relating them to the usual connectives from dependent type theory and linear logic. Finally, one important class of models, given by families with values in some symmetric monoidal category, is investigated in detail.
\end{abstract}
\vspace{120pt}
\quad
\clearpage \noindent \textbf{Disclaimer and Acknowledgements}\\
\noindent The concept of a syntax with linear dependent types is not new. Such a calculus was first considered in \cite{cervesato1996linear}, in the context of a linear extension of the Logical Framework (LF). In terms of the semantics, the author would like to point to \cite{shulman2013enriched} and \cite{ponto2012duality}, which study very similar semantic objects, although without the notion of comprehension. Finally, the author would like to thank Urs Schreiber for sparking his interest in the topic through many enthusiastic posts on the nLab and nForum.\\
\\
The contributions of the present work are as follows.
\begin{enumerate}
\item The presentation of a syntax in a style that is very close to both the dual intuitionistic linear logic of \cite{barber1996dual} and intuitionistic dependent type theory as presented in \cite{hofmann1997syntax}. This clarifies, from a syntactic point of view, how exactly linear dependent types fit in with the work in both traditions.
\item The addition of various structural rules that allow us to consider a more basic setting without the various type formers of the Linear Logical Framework (LLF). In the rich setting of LLF, these become admissible.
\item The addition of various type formers to the syntax (most notably, $\Sigma$-, $!$-, and $\mathrm{Id}$-types).
\item The development of the first categorical semantics for linear dependent types: although there have been some suggestions about this on the nLab, in particular by Mike Shulman and Urs Schreiber, no account appears to have been published, as far as the author is aware. The semantics is developed in a style that combines features of the linear-non-linear adjunctions of \cite{barber1996dual} and the comprehension categories of \cite{jacobs1993comprehension}. Among other things, it shows that multiplicative quantifiers arise naturally, as adjoints to substitution, thereby offering a new point of view on the long-debated issue of quantifiers in linear logic.
\end{enumerate}

\clearpage
\tableofcontents
\clearpage

\section*{Introduction}
To put this work in context, the best point of departure may be Church's simply typed $\lambda$-calculus (or intuitionistic propositional type theory) \cite{church1940formulation}, which, according to the Curry-Howard correspondence (e.g. \cite{howard1995formulae}), can be thought of as a proof calculus for intuitionistic propositional logic. To this, Lambek (e.g. \cite{lambek1972deductive}) added the new idea that it can also be viewed as a syntax for describing Cartesian closed categories. From this starting point, two traditions branch off that are particularly relevant for us.

On the one hand, so-called dependent type theory can be seen to extend this, along the Curry-Howard correspondence, to provide a proof calculus for intuitionistic predicate logic. See e.g. \cite{martin1998intuitionistic}. Various flavours of categorical semantics for this calculus have been given, but they almost all amount to the same essential ingredients: certain fibred or, equivalently, indexed Cartesian categories, usually (depending on the exact formulation of the type theory) with a notion of comprehension that relates the fibres to the base category.

On the other hand, there is so-called linear logic, a school of logic initiated by Girard, which has pursued a further, resource-sensitive analysis of the intuitionistic propositional type theory by exposing precisely how many times each assumption is used in proofs. See e.g. \cite{girard1987linear}. Later, the essence of Girard's system was seen to be captured by a less restrictive system called (multiplicative) intuitionistic linear logic (where Girard's original system is now referred to as classical linear logic). This will be the sense in which we use the word 'linear'. A satisfactory proof calculus and categorical semantics, in terms of symmetric monoidal closed categories, are easily given. The final essential element of linear type theories, the resource modality providing the connection with intuitionistic type theories, can be given categorical semantics as a monoidal adjunction between an intuitionistic and a linear model \cite{barber1996dual}.

The question arises whether this linear analysis can be extended to predicate logic. Although some work has been done in this direction, the author feels a satisfactory answer has not been given.

Although Girard's early work in linear logic already discussed quantifiers, the analysis appears to have stayed rather superficial. In particular, an account of internal quantification, or a linear variant of Martin-L\"of's type theory, was missing, let alone a Curry-Howard correspondence. Later, linear types and dependent types were first combined in \cite{cervesato1996linear}, where a syntax was presented that extends LF (Logical Framework) with linear types (that depend on terms of intuitionistic types). This has given rise to a line of work in the computer science community. See e.g. \cite{dal2011linear,petit2012linear,gaboardi2013linear}. All the work seems to be syntactic in nature, however.

On the other hand, very similar ideas, this time at the level of categorical semantics and specific models (coming from homotopy theory, algebra, and mathematical physics), have recently emerged in the mathematical community, perhaps independently. Relevant literature includes \cite{may2006parametrized,shulman2013enriched,ponto2012duality,schreiber2014quantization}.

Although some suggestions about possible connections between the two lines of work have been made on the nLab and nForum in recent months, no account of the correspondence has been published, as far as the author is aware. Moreover, the syntactic tradition seems to have stayed restricted to the situation of functional type theory, in which one has $\Pi$- and $\multimap$-types, while our interest is in the general case of algebraic type theory, which is of fundamental importance from the point of view of mathematics, where structures seldom admit internal homs. At the same time, the syntactic tradition seems to lack other type formers (such as $\Sigma$-, $!$-, and $\mathrm{Id}$-types). The semantic tradition, however, does not seem to have given a sufficient account of the notion of comprehension. The present text takes some steps to close this gap in the existing literature.

The point of this paper is to illustrate how linear and dependent types can be combined straightforwardly and in great generality. First, in section \ref{sec:lindep}, we start with a general discussion of combining linear and dependent types. Second, in section \ref{sec:syn}, we present a syntax, intuitionistic linear dependent type theory (ILDTT), a natural blend of the dual intuitionistic linear logic (DILL) of \cite{barber1996dual} and dependent type theory (DTT), as presented in, e.g., \cite{hofmann1997syntax}. Third, in section \ref{sec:sem}, we present a complete categorical semantics that is the obvious combination of linear/non-linear adjunctions (e.g. \cite{barber1996dual}) and indexed Cartesian categories with comprehension (e.g. \cite{jacobs1993comprehension}). Finally, in section \ref{sec:dismod}, an important class of (rather discrete) models is investigated in terms of families with values in a fixed symmetric monoidal category.

This paper is the first report of a research programme that aims to explore the interplay between type dependency and linearity. We will therefore end with a brief discussion of some of our work in progress and planned work on related topics.
\clearpage
\section{Preliminaries}
We start by suggesting references on both linear type theory and dependent type theory, since having the right point of view on both subjects greatly simplifies understanding our presentation of the new material.

\subsection{(Non-Linear) Intuitionistic Dependent Type Theory (DTT)}
\subsubsection*{Syntax}
For an introductory but thorough treatment of the syntax, we refer the reader to \cite{hofmann1997syntax}.

\subsubsection*{Categorical Semantics} There are many (more or less) equivalent kinds of categorical models of a dependent type theory. Perhaps the best notion to have in mind while reading this text is the notion of a full (split) comprehension category of \cite{jacobs1993comprehension} or, equivalently, the following rephrasing of a full split comprehension category with $1$- and $\times$-types\footnote{This restriction avoids our having to talk about multicategories. A reader familiar with multicategories will see the obvious generalisation to the case without $1$- and $\times$-types.}. 

\begin{definition*}[Strict Indexed Cartesian Monoidal Category with Comprehension] By a strict\footnote{The strictness here refers to the fact that $\mathcal{I}$ is a functor, rather than a pseudofunctor. This reflects the fact that term substitution in types is a strict operation in type theory, even though the non-strict version might seem more natural from the point of view of category theory. It should be noted, though, that every indexed category is equivalent to a strict one. \cite{power1989general}} indexed\footnote{Of course, using the Grothendieck construction, we could equivalently use split fibred categories here, as is done in certain accounts of semantics for DTT. However, indexed categories seem closer to the syntax and are technically less involved in the strict case.} Cartesian monoidal category with comprehension, we mean the following data.
\begin{enumerate}
\item A category of intuitionistic contexts $\mathcal{C}$ with a terminal object $\cdot$.
\item A strict indexed Cartesian monoidal category $\mathcal{I}$ over $\mathcal{C}$, i.e. a contravariant functor $\mathcal{I}$ into the category $\mathrm{CMCat}$ of (small) Cartesian monoidal categories and strong cartesian monoidal functors $$\mathcal{C}^{op}\ra{\mathcal{I}}\mathrm{CMCat}.$$
We will also write $-\{f\}:=\mathcal{I}(f)$ for the action of $\mathcal{I}$ on a morphism $f$ of $\mathcal{C}$.
\item A \emph{comprehension schema}, i.e. for each $\Delta\in\mathrm{ob}(\mathcal{C})$ and $A\in\mathrm{ob}(\mathcal{I}(\Delta))$ a representation for the functor $$x\mapsto\mathcal{I}(\mathrm{dom}(x))(1,A\{x\}):(\mathcal{C}/\Delta)^{op}\ra{}\mathrm{Set},$$
which induces a fully faithful \emph{comprehension functor}: if we write the representing object $\Delta.{A}\ra{\mathbf{p}_{\Delta,{A}}}\Delta\in\mathrm{ob}(\mathcal{C}/\Delta)$, write $a  \mapsto  \langle f,a\rangle$ for the isomorphism $\mathcal{I}(\Delta')(1,A\{f\}) \cong  \mathcal{C}/\Delta(f,\mathbf{p}_{\Delta,{A}})$, and write $\mathbf{v}_{\Delta,{A}}\in\mathcal{I}(\Delta.{A})(1,A\{\mathbf{p}_{\Delta,{A}}\})$ for the universal element, the comprehension functor $\mathcal{I}\ra{}\mathcal{C}/-$ is defined as
\[
\begin{tikzcd}[ampersand replacement=\&, column sep=large]
\mathcal{I}(\Delta) \arrow[r] \& \mathcal{C}/\Delta\\
A\ra{f}B \arrow[r, maps to] \& \mathbf{p}_{\Delta,A} \arrow[r,"{\langle \mathbf{p}_{\Delta,A},\mathcal{I}(\mathbf{p}_{\Delta,A})(f)\circ \mathbf{v}_{\Delta,A}\rangle}"] \& \mathbf{p}_{\Delta,B}.
\end{tikzcd}
\]
\end{enumerate}
\end{definition*}
Note that the comprehension schema says that we build up the morphisms into objects that arise as lists of types in our category of contexts $\mathcal{C}$ as lists of closed terms. This is the crucial aspect that ensures that we are in the world of internal quantification. The fact that we require the comprehension functor to be fully faithful means that the non-closed terms in $\mathcal{I}(\Delta)(A,B)$ correspond precisely to the closed terms $\mathcal{I}(\Delta.A)(I,B)$. The equivalent condition on a comprehension category is the notion of fullness of \cite{jacobs1993comprehension}. This is essential to get a precise fit with the syntax. However, we will drop this restriction when modelling linear dependent types, since having $A$ as a linear assumption (an assumption in the fibre) can then genuinely differ from having $A$ as an intuitionistic assumption (an assumption in the base).

\subsection{Intuitionistic Linear (Non-Dependent) Type Theory (ILTT)}
The most suitable flavour of ILTT to have in mind while reading these notes is the so-called dual intuitionistic linear logic (DILL) of \cite{barber1996dual}. This will be our principal reference for both syntax and semantics. Both the syntax and semantics given in this reference are very close in spirit to our syntax and semantics of linear dependent types. For the semantics, more background is provided in \cite{mellies2009categorical}. 

\clearpage
\section{Intuitionistic Linear Dependent Type Theory?}
\label{sec:lindep}
Although it is a priori not entirely clear what linear dependent type theory should be, one can identify some guiding criteria. My initial reaction was the following.

The goal is a version of dependent type theory without weakening and contraction rules, together with an exponential co-modality that restores the missing structural rules.

When one first tries to write down such a type theory, however, one runs into the following discrepancy.
\begin{itemize}
\item The lack of weakening and contraction rules in \emph{linear type theory} forces us to refer to each declared variable precisely once: for a sequent $x:A\vdash t:B$, we know that $x$ has a unique occurrence in $t$.
\item In \emph{dependent type theory}, types can have free (term) variables: $x:A\vdash B\;\mathrm{type}$, where $x$ is a free variable in $B$. Crucially, we can then talk about terms of $B$: $x:A\vdash b:B$, where, generally, $x$ may also be free in $b$. For almost all interesting applications we will need multiple occurrences of $x$ to construct $b:B$, at least one for $B$ and one for $b$.
\end{itemize}
The question now is what it means to refer to a declared variable only once.

Do we not count occurrences in types? This point of view seems incompatible with universes, however, which play an important role in dependent type theory. If we do count them, however, the language seems to lose much of its expressive power. In particular, it seems to prevent us from talking about constant types.\\
\\
In this paper, we will circumvent the issue \emph{by restricting to type dependency on terms of intuitionistic types}. In this case, there is no conflict, as those terms can be copied and deleted freely.

A semantic, rather than syntactic, argument for this system is the following. We will see that if we start out with a model of linear type theory with external quantification (i.e. a strict indexed symmetric monoidal category) and demand that the quantification be internal (i.e. demand that it be a linear dependent type theory) in the natural way (i.e. impose the comprehension axiom), then our base category becomes a cartesian category (i.e. a model of intuitionistic type theory).

Our best attempt at giving meaning to type dependency on linear types will be through dependency on a (co-Kleisli or co-Eilenberg-Moore) category of co-algebras for $!$. We will briefly return to the issue later, but most of the discussion is beyond the scope of this paper.\\
\\
In a linear dependent type theory, one could initially conceive of an additive $\Sigma$-type, denoted by $\Sigma^\&_{x:X}A$, and a multiplicative $\Sigma$-type, denoted by $\Sigma^\otimes_{x:X}A$, to generalise the additive and multiplicative conjunction of intuitionistic linear logic. The modality would relate the additive $\Sigma$-type of linear type families to the multiplicative $\Sigma$-type of intuitionistic type families, through a Seely-like isomorphism:
$$!(\Sigma^\&_{x:X}A)=\Sigma^\otimes_{x:!X}!A.$$
The idea is that the linear $\Sigma$-types should simultaneously generalise the conjunctions from propositional linear logic and the $\Sigma$-types from intuitionistic dependent type theory.

In general, though, since we are restricting to type dependency on $!$-ed types, the additive $\Sigma$-type (as far as one can make sense of the notion) would become effectively intuitionistic (as both terms in the introduction rule would have to depend on the same context, one of which is forced to be intuitionistic). This also breaks the symmetry in the Seely isomorphism. Moreover, we will see that the multiplicative $\Sigma$-types (and $\Pi$-types) arise naturally in the categorical semantics. We will later briefly return to the possibility of a genuinely linear additive $\Sigma$-type and Seely isomorphisms.

For the $\Pi$-type, one would also expect to obtain two linear analogues: an additive and a multiplicative version. However, given the overwhelming preference in the linear logic literature for multiplicative implication, we can probably safely restrict our attention to multiplicative $\Pi$-types for now.
\clearpage
\section{Syntax of ILDTT}
\label{sec:syn}
We assume the reader has some familiarity with the formal syntax of dependent type theory and linear type theory. In particular, we will not go into syntactic details such as $\alpha$-conversion, name-binding, capture-free substitution of $a$ for $x$ in $t$ (write $t[a/x]$), and pre-syntax. The reader can find details on all of these topics in \cite{hofmann1997syntax}.

We next present the formal syntax of ILDTT. We start with a presentation of the judgements that represent the propositions in the language and then discuss the rules of inference: first the structural core, then the logical rules for a series of optional type formers. We conclude this section with a few basic results about the syntax.

\subsubsection*{Judgements}
We adopt a notation $\Delta;\Xi$ for contexts, where $\Delta$ is `an intuitionistic region' and $\Xi$ is `a linear region', as in \cite{barber1996dual}. The idea is that we have an empty context and can extend an existing context $\Delta;\Xi$ with both intuitionistic and linear types that are allowed to depend on $\Delta$.

Our language will express judgements of the following six forms.

\begin{figure}[h]
\centering
\fbox{\parbox{\textwidth}{\footnotesize
\begin{tabular}{ll}
\textbf{ILDTT judgement} & \textbf{Intended meaning}\vspace{2pt}\\
$\vdash \Delta;\Xi \;\mathrm{ctxt}$ & $\Delta;\Xi$ is a valid context\\
$\Delta;\cdot \vdash A\;\mathrm{type}$ &  $A$ is a type in (intuitionistic) context $\Delta$\\
$\Delta;\Xi\vdash a:A$ & $a$ is a term of type $A$ in context $\Delta;\Xi$\\
$\vdash \Delta;\Xi \equiv \Delta';\Xi'\;\mathrm{ctxt}$\hspace{40pt} & $\Delta;\Xi$ and $\Delta';\Xi'$ are judgementally equal contexts\\
$\Delta;\cdot\vdash A\equiv A'\;\mathrm{type}$ & $A$ and $A'$ are judgementally equal types in (intuitionistic) context $\Delta$\\
$\Delta;\Xi\vdash a\equiv a':A$ & $a$ and $a'$ are judgementally equal terms of type $A$ in context $\Delta;\Xi$
\end{tabular}}}
\normalsize
\caption{Judgements of ILDTT.}
\end{figure}

\subsubsection*{Structural Rules}
We will use the following structural rules, which are essentially the structural rules of dependent type theory where some rules appear in both an intuitionistic and a linear form. We present the rules by group, with their names, from left to right and top to bottom.

\begin{figure}[h]
\centering
\fbox{\parbox{\textwidth}{\footnotesize
\quad\\
\quad\\
\begin{tabular}{lr}
\AxiomC{}
\RightLabel{C-Emp}
\UnaryInfC{$\vdash \cdot;\cdot \;\mathrm{ctxt}$}
\DisplayProof
& \\
& \\
\AxiomC{$\vdash\Delta;\Xi\; \mathrm{ctxt}$}
\AxiomC{$\Delta;\cdot \vdash A\;\mathrm{type}$}
\RightLabel{Int-C-Ext}
\BinaryInfC{$\vdash \Delta,x:A;\Xi \;\mathrm{ctxt}$}
\DisplayProof

&
\hspace{56pt}
\AxiomC{$\vdash \Delta;\Xi\equiv\Delta';\Xi'\;\mathrm{ctxt}$}
\AxiomC{$\Delta;\cdot \vdash A\equiv B\;\mathrm{type}$}
\RightLabel{Int-C-Ext-Eq}
\BinaryInfC{$\vdash \Delta,x:A;\Xi\equiv\Delta',y:B;\Xi'\;\mathrm{ctxt}$}
\DisplayProof\\
& \\
\AxiomC{$\vdash \Delta;\Xi\;\mathrm{ctxt}$}
\AxiomC{$\Delta;\cdot \vdash A\;\mathrm{type}$}
\RightLabel{Lin-C-Ext}
\BinaryInfC{$\vdash \Delta;\Xi,x:A\;\mathrm{ctxt}$}
\DisplayProof

&

\AxiomC{$\vdash \Delta;\Xi\equiv\Delta';\Xi'\;\mathrm{ctxt}$}
\AxiomC{$\Delta;\cdot\vdash A\equiv B\;\mathrm{type}$}
\RightLabel{Lin-C-Ext-Eq}
\BinaryInfC{$\vdash \Delta;\Xi,x:A\equiv\Delta';\Xi',y:B\;\mathrm{ctxt}$}
\DisplayProof
\\
&\\
&\\
\AxiomC{$\vdash \Delta,x:A,\Delta';\cdot\;\mathrm{ctxt}$}
\RightLabel{Int-Var}
\UnaryInfC{$\Delta,x:A,\Delta';\cdot\vdash x:A$}
\DisplayProof
&\hspace{174pt}
\AxiomC{$\vdash \Delta;x:A\;\mathrm{ctxt}$}
\RightLabel{Lin-Var}
\UnaryInfC{$\Delta;x:A\vdash x:A$}
\DisplayProof
\end{tabular}}}
\normalsize
\caption{Context formation and variable declaration rules.}
\end{figure}

\vspace{-30pt}
\begin{figure}[h]
\centering
\fbox{\parbox{\textwidth}{\footnotesize
\quad\\
\begin{tabular}{lr}
\AxiomC{$\vdash \Delta;\Xi\;\mathrm{ctxt}$}
\RightLabel{C-Eq-R}
\UnaryInfC{$\vdash \Delta;\Xi\equiv \Delta;\Xi\;\mathrm{ctxt}$}
\DisplayProof
&\hspace{103pt}
\AxiomC{$\vdash \Delta;\Xi\equiv \Delta';\Xi'\;\mathrm{ctxt}$}
\RightLabel{C-Eq-S}
\UnaryInfC{$\vdash \Delta';\Xi'\equiv \Delta;\Xi\;\mathrm{ctxt}$}
\DisplayProof\\
&\\
\AxiomC{$\vdash \Delta;\Xi\equiv \Delta';\Xi'\;\mathrm{ctxt}$}
\AxiomC{$\vdash \Delta';\Xi'\equiv \Delta'';\Xi''\;\mathrm{ctxt}$}
\RightLabel{C-Eq-T}
\BinaryInfC{$\vdash \Delta;\Xi\equiv \Delta'';\Xi''\;\mathrm{ctxt}$}
\DisplayProof &\\
&\\
\AxiomC{$\Delta;\cdot\vdash A\;\mathrm{type}$}
\RightLabel{Ty-Eq-R}
\UnaryInfC{$\Delta;\cdot\vdash A\equiv A\;\mathrm{type}$}
\DisplayProof
&
\AxiomC{$\Delta;\cdot\vdash A\equiv A'\;\mathrm{type}$}
\RightLabel{Ty-Eq-S}
\UnaryInfC{$\Delta;\cdot\vdash A'\equiv A\;\mathrm{type}$}
\DisplayProof\\
&\\
\AxiomC{$\Delta;\cdot\vdash A\equiv A'\;\mathrm{type}$}
\AxiomC{$\Delta;\cdot\vdash A'\equiv A''\;\mathrm{type}$}
\RightLabel{Ty-Eq-T}
\BinaryInfC{$\Delta;\cdot\vdash A\equiv A''\;\mathrm{type}$}
\DisplayProof
&\\
&\\
\AxiomC{$\Delta;\Xi\vdash a:A$}
\RightLabel{Tm-Eq-R}
\UnaryInfC{$\Delta;\Xi\vdash a\equiv a: A$}
\DisplayProof
&
\AxiomC{$\Delta;\Xi\vdash a\equiv a':A$}
\RightLabel{Tm-Eq-S}
\UnaryInfC{$\Delta;\Xi\vdash a'\equiv a: A$}
\DisplayProof
\\
&\\
\AxiomC{$\Delta;\Xi\vdash a\equiv a':A$}
\AxiomC{$\Delta;\Xi\vdash a'\equiv a'':A$}
\RightLabel{Tm-Eq-T}
\BinaryInfC{$\Delta;\Xi\vdash a\equiv a'': A$}
\DisplayProof
\end{tabular}
\\
\\
\\
\begin{tabular}{ll}
\AxiomC{$\Delta;\Xi\vdash a:A$}
\AxiomC{$\vdash \Delta;\Xi\equiv \Delta';\Xi'\;\mathrm{ctxt}$}
\AxiomC{$\Delta;\cdot \vdash A\equiv A'\;\mathrm{type}$}
\RightLabel{Tm-Conv}
\TrinaryInfC{$\Delta';\Xi'\vdash a:A'$}
\DisplayProof &\\
&\\
\AxiomC{$\Delta;\cdot\vdash A\;\mathrm{type}$}
\AxiomC{$\vdash \Delta;\cdot\equiv \Delta';\cdot\;\mathrm{ctxt}$}
\RightLabel{Ty-Conv}
\BinaryInfC{$\Delta';\cdot\vdash A\;\mathrm{type}$}
\DisplayProof
\end{tabular}
\normalsize}}
\caption{A few standard rules for judgemental equality, saying that it is an equivalence relation and is compatible with typing.}
\end{figure}
\begin{figure}[h]
\centering
\fbox{\parbox{\textwidth}{\footnotesize
\quad\\
\begin{tabular}{lr}
\AxiomC{$\Delta,\Delta';\Xi\vdash\mathcal{J}$}
\AxiomC{$\Delta;\cdot\vdash A\;\mathrm{type}$}
\RightLabel{Int-Weak}
\BinaryInfC{$\Delta,x:A,\Delta';\Xi\vdash \mathcal{J}$}
\DisplayProof
&\\
&\\
&\\
\AxiomC{$\Delta,x:A,x':A',\Delta';\Xi\vdash \mathcal{J}$}
\RightLabel{Int-Exch}
\UnaryInfC{$\Delta,x':A',x:A,\Delta';\Xi\vdash \mathcal{J}$}
\DisplayProof
&\vspace{4pt}
\AxiomC{$\Delta;\Xi,x:A,x':A',\Xi'\vdash \mathcal{J}$}
\RightLabel{Lin-Exch}
\UnaryInfC{$\Delta;\Xi,x':A',x:A,\Xi'\vdash \mathcal{J}$}
\DisplayProof
\\
(if $x$ is not free in $A'$)&\\
&\\
\AxiomC{${\Delta},x:A,\Delta';\cdot \vdash B\;\mathrm{type}$}
\AxiomC{$\Delta;\cdot \vdash a:A$}
\RightLabel{Int-Ty-Subst}
\BinaryInfC{${\Delta},\Delta'[{a}/x];\cdot \vdash B[{a}/x]\;\mathrm{type}$}
\DisplayProof\hspace{-12pt}
&
\AxiomC{${\Delta},x:A,\Delta';\cdot \vdash B\equiv B'\;\mathrm{type}$}
\AxiomC{$\Delta;\cdot \vdash a:A$}
\RightLabel{Int-Ty-Subst-Eq}
\BinaryInfC{${\Delta},\Delta'[{a}/x];\cdot \vdash B[{a}/x]\equiv B'[{a}/x]\;\mathrm{type}$}
\DisplayProof\\
&\\
&\\
\AxiomC{${\Delta},x:A,\Delta';\Xi \vdash b:B$}
\AxiomC{$\Delta;\cdot \vdash a:A$}
\RightLabel{Int-Tm-Subst}
\BinaryInfC{${\Delta},\Delta'[{a}/x];\Xi[{a}/x] \vdash b[{a}/x]:B[{a}/x]$}
\DisplayProof\hspace{-9pt}
&
\AxiomC{${\Delta},x:A,\Delta';\Xi \vdash b\equiv b':B$}
\AxiomC{$\Delta;\cdot \vdash a:A$}
\RightLabel{Int-Tm-Subst-Eq}
\BinaryInfC{${\Delta},\Delta'[{a}/x];\Xi[{a}/x] \vdash b[{a}/x]\equiv b'[{a}/x]:B[{a}/x]$}
\DisplayProof
\\
&\\
&\\
\AxiomC{$\Delta;\Xi,x:A\vdash b:B$}
\AxiomC{$\Delta;\Xi'\vdash a:A$}
\RightLabel{Lin-Tm-Subst}
\BinaryInfC{$\Delta;\Xi,\Xi'\vdash b[a/x]:B$}
\DisplayProof
&
\AxiomC{$\Delta;\Xi,x:A\vdash b\equiv b':B$}
\AxiomC{$\Delta;\Xi'\vdash a:A$}
\RightLabel{Lin-Tm-Subst-Eq}
\BinaryInfC{$\Delta;\Xi,\Xi'\vdash b[a/x]\equiv b'[a/x]:B$}
\DisplayProof\\
&\\
&\\
\end{tabular}\\
\vspace{-15pt}
\normalsize}}
\caption{Exchange, weakening, and substitution rules. Here, $\mathcal{J}$ represents a statement of the form $B\;\mathrm{type}$, $B\equiv B'$, $b:B$, or $b\equiv b':B$, such that all judgements are well-formed.}
\end{figure}
\vspace{-30pt}
\clearpage

\subsubsection*{Logical Rules} We introduce some basic (optional) type and term formers, for which we will give type formation (denoted -F), term introduction (-I), term elimination (-E), term computation rules (-C), and (judgemental) term uniqueness principles (-U). We also assume the obvious rules stating that the new type formers and term formers respect judgemental equality. Moreover, $\Sigma_{!x:!A}$, $\Pi_{!x:!A}$, $\lambda_{!x:!A}$, and $\lambda_{x:A}$ are name-binding operators, binding free occurrences of $x$ within their scope. Anticipating some theorems of the calculus, we overload some of the notation for the -I and -E rules for various type formers, in order to avoid syntactic clutter. Uniqueness of typing can easily be restored by carrying around enough type information in the notation corresponding to the various -I and -E rules.

We require -U-rules for the various type formers in this paper, as this allows us to give a natural categorical semantics. In practice, when building a computational implementation of a type theory like ours, one would probably drop these rules to make the system decidable, which would correspond to switching to weak variants of the categorical constructions presented here.\footnote{In that case, in DTT, one would usually require some stronger `dependent' elimination rules, which would make propositional equivalents of the -U-rules provable, adding some extensionality to the system, while preserving its computational properties. Such rules are problematic in ILDTT, however, from both syntactic and semantic points of view, and further investigation is warranted here.}

\begin{figure}[h]
\centering
\fbox{
\parbox{\textwidth}{\footnotesize
\begin{tabular}{lr}
\AxiomC{$\Delta,x:A;\cdot\vdash B\;\mathrm{type}$}
\RightLabel{$\Sigma$-F}
\UnaryInfC{$\Delta;\cdot\vdash \Sigma_{!x:!{A}}B\;\mathrm{type}$}
\DisplayProof
&

\\
\AxiomC{$\Delta;\cdot \vdash a:A$}
\AxiomC{$\Delta;\Xi \vdash b:B[{a}/x]$}
\RightLabel{$\Sigma$-I}
\BinaryInfC{$\Delta ; \Xi \vdash ! {a} \otimes b:\Sigma_{!x:!{A}}B $}
\DisplayProof
&
\AxiomC{$\Delta;\cdot \vdash C\;\mathrm{type}$}
\noLine
\UnaryInfC{$\Delta;\Xi \vdash t:\Sigma_{!x:!{A}}B$}
\noLine
\UnaryInfC{$\Delta,x:A;\Xi',y:B\vdash c:C$}
\RightLabel{$\Sigma$-E}
\UnaryInfC{$\Delta;\Xi,\Xi' \vdash \mathrm{let}\;t\;\mathrm{be}\;  !{x} \otimes y \;\mathrm{in}\;c:C$}
\DisplayProof
\\
\\
\AxiomC{$ \Delta;\Xi\vdash \mathrm{let}\; !{a} \otimes b\;\mathrm{be}\; ! {x} \otimes y \;\mathrm{in}\;c:C$}
\RightLabel{$\Sigma$-C}
\UnaryInfC{$\Delta;\Xi\vdash \mathrm{let}\; !{a} \otimes b\;\mathrm{be}\; ! {x} \otimes y \;\mathrm{in}\;c\equiv c[{a}/x,b/y]:C$}
\DisplayProof\hspace{44pt}
&

\AxiomC{$ \Delta;\Xi\vdash \mathrm{let}\;t\;\mathrm{be}\; ! {x} \otimes y \;\mathrm{in}\; !{x} \otimes y:\Sigma_{!x:!{A}}B$}
\RightLabel{$\Sigma$-U}
\UnaryInfC{$\Delta;\Xi\vdash \mathrm{let}\;t\;\mathrm{be}\;  !{x} \otimes y \;\mathrm{in}\; !{x} \otimes y\equiv t:\Sigma_{!x:!{A}}B$}
\DisplayProof
\\
&\\
&\\
&\\
\AxiomC{$\Delta,x:A;\cdot \vdash B\;\mathrm{type}$}
\RightLabel{$\Pi$-F}
\UnaryInfC{$\Delta;\cdot\vdash\Pi_{!x:!{A}}B\;\mathrm{type}$}
\DisplayProof
&\\
&\\
\AxiomC{$\vdash \Delta;\Xi\;\mathrm{ctxt}$}
\AxiomC{$\Delta,x:A;\Xi\vdash b:B$}
\RightLabel{$\Pi$-I}
\BinaryInfC{$\Delta;\Xi\vdash \lambda_{!x:{!A}}b:\Pi_{!x:{!A}}B$}
\DisplayProof
&
\AxiomC{$\Delta;\cdot \vdash a:A$}
\AxiomC{$\Delta;\Xi\vdash f:\Pi_{!x:!{A}}B$}
\RightLabel{$\Pi$-E}
\BinaryInfC{$\Delta;\Xi\vdash f(!{a}):B[{a}/x]$}
\DisplayProof\\
&\\
&\\
\AxiomC{$\Delta;\Xi\vdash (\lambda_{!x:!{A}}b)(!{a}):B[{a}/x]$}
\RightLabel{$\Pi$-C}
\UnaryInfC{$\Delta;\Xi\vdash (\lambda_{!x:!{A}}b)(!{a})\equiv b[{a}/x]:B[{a}/x]$}
\DisplayProof
&
\AxiomC{$\Delta;\Xi\vdash \lambda_{!x:!{A}}f(!x):\Pi_{!x:!{A}}B$}
\RightLabel{$\Pi$-U}
\UnaryInfC{$\Delta;\Xi\vdash f\equiv \lambda_{!x:{!A}}f(!x):\Pi_{!x:{!A}}B$}
\DisplayProof
\\
&\\
&\\
&\\
\begin{tabular}{l}
\AxiomC{$\Delta;\cdot \vdash a:A$}
\AxiomC{$\Delta;\cdot \vdash a':A$}
\RightLabel{$\mathrm{Id}$-F}
\BinaryInfC{$\Delta;\cdot \vdash \mathrm{Id}_{!A}(a,a')\;\mathrm{type}$}
\DisplayProof \\
\\
\AxiomC{$\Delta;\cdot \vdash a:A$}
\RightLabel{$\mathrm{Id}$-I}
\UnaryInfC{$\Delta;\cdot \vdash \mathrm{refl}_{!a}:\mathrm{Id}_{!A}(a,a)$}
\DisplayProof
\end{tabular}\hspace{-64pt}
&\hspace{-64pt}\begin{tabular}{l}
\hspace{77pt}$\Delta,x:A,x':A;\cdot \vdash D\;\mathrm{type}$\\
\hspace{77pt}$\Delta,z:A;\Xi \vdash d:D[{z}/x,{z}/x']$\\
\hspace{77pt}$\Delta ;\cdot \vdash a:A$\\
\hspace{77pt}$\Delta;\cdot \vdash a':A$\\
\AxiomC{$\Delta ;\Xi' \vdash p:\mathrm{Id}_{!A}(a,a')$}
\RightLabel{$\mathrm{Id}$-E}
\UnaryInfC{$\Delta;\Xi[a/z],\Xi' \vdash \mathrm{let}\; (a,a',p)\;\mathrm{be}\;(z,z,\mathrm{refl}_{!z})\;\mathrm{in}\; d:D[{a}/x,{a'}/x']$}
\DisplayProof
\end{tabular}\hspace{-5pt}
\\
&\\
&\\
\AxiomC{$\Delta;\Xi\vdash \mathrm{let}\; (a,a,\mathrm{refl}_{!a})\;\mathrm{be}\;(z,z,\mathrm{refl}_{!z})\;\mathrm{in}\; d:D[{a}/x,{a}/x']$}
\RightLabel{$\mathrm{Id}$-C}
\UnaryInfC{$\Delta;\Xi\vdash \mathrm{let}\; (a,a,\mathrm{refl}_{!a})\;\mathrm{be}\;(z,z,\mathrm{refl}_{!z})\;\mathrm{in}\; d\equiv d[{a}/z] :D[{a}/x,{a}/x']$}
\DisplayProof \hspace{-46pt}
& \hspace{-46pt} \\
\end{tabular}\vspace{10pt}
\begin{tabular}{lr}
\AxiomC{$\Delta,x:A,x':A;\Xi, z:\mathrm{Id}_{!A}(x,x') \vdash \mathrm{let}\;(x,x',z)\;\mathrm{be}\;(x,x,\mathrm{refl}_{!x})\;\mathrm{in}\;c[x/x',\mathrm{refl}_{!x}/z]:C$}
\RightLabel{$\mathrm{Id}$-U}
\UnaryInfC{$\Delta,x:A,x':A;\Xi, z:\mathrm{Id}_{!A}(x,x') \vdash \mathrm{let}\;(x,x',z)\;\mathrm{be}\;(x,x,\mathrm{refl}_{!x})\;\mathrm{in}\;c[x/x',\mathrm{refl}_{!x}/z]\equiv c:C$}
\DisplayProof \hspace{-176pt}&\hspace{-176pt}
\end{tabular}
\normalsize
}}
\caption{Rules for linear equivalents of some of the usual type formers from DTT: $\Sigma$-, $\Pi$-, and $\mathrm{Id}$-types.}
\end{figure}
\nopagebreak
\begin{figure}[h]
\centering
\fbox{\parbox{\textwidth}{\footnotesize 
\quad\\
\begin{tabular}{lr}
\AxiomC{}
\RightLabel{$I$-F}
\UnaryInfC{$\Delta;\cdot\vdash I\;\mathrm{type}$}
\DisplayProof
&\\
\AxiomC{}
\RightLabel{$I$-I}
\UnaryInfC{$\Delta;\cdot\vdash *:I$}
\DisplayProof
&
\AxiomC{$\Delta;\Xi'\vdash t:I$}
\AxiomC{$\Delta;\Xi\vdash a:A$}
\RightLabel{$I$-E}
\BinaryInfC{$\Delta;\Xi,\Xi'\vdash \mathrm{let}\;t\;\mathrm{be}\;*\;\mathrm{in}\;a:A$}
\DisplayProof 
 \\
&\\
\AxiomC{$\Delta;\Xi\vdash  \mathrm{let}\;*\;\mathrm{be}\;*\;\mathrm{in}\;a:A$}
\RightLabel{$I$-C}
\UnaryInfC{$\Delta;\Xi\vdash \mathrm{let}\;*\;\mathrm{be}\;*\;\mathrm{in}\;a\equiv a :A$}
\DisplayProof
&
\AxiomC{$\Delta;\Xi\vdash \mathrm{let}\;t\;\mathrm{be}\;*\;\mathrm{in}\;*:I$}
\RightLabel{$I$-U}
\UnaryInfC{$\Delta;\Xi\vdash \mathrm{let}\;t\;\mathrm{be}\;*\;\mathrm{in}\;*\equiv t :I$}
\DisplayProof 
\\
&\\
\AxiomC{$\Delta;\cdot \vdash A\;\mathrm{type}$}
\AxiomC{$\Delta;\cdot \vdash B\;\mathrm{type}$}
\RightLabel{$\otimes$-F}
\BinaryInfC{$\Delta;\cdot \vdash A\otimes B\;\mathrm{type}$}
\DisplayProof 
&\\
&\\
\AxiomC{$\Delta;\Xi\vdash a:A$}
\AxiomC{$\Delta;\Xi'\vdash b:B$}
\RightLabel{$\otimes$-I}
\BinaryInfC{$\Delta;\Xi,\Xi'\vdash a\otimes b:A\otimes B$}
\DisplayProof
&
\AxiomC{$\Delta;\Xi\vdash t:A\otimes B$}
\AxiomC{$\Delta;\Xi',x:A,y:B\vdash c:C$}
\RightLabel{$\otimes$-E}
\BinaryInfC{$\Delta;\Xi,\Xi'\vdash \mathrm{let}\; t\;\mathrm{be}\;x\otimes y\;\mathrm{in} \; c:C$}
\DisplayProof
\\
&\\
\AxiomC{$\Delta;\Xi \vdash \mathrm{let}\; a\otimes b \;\mathrm{be}\;x\otimes y\;\mathrm{in} \; c :C$}
\RightLabel{$\otimes$-C}
\UnaryInfC{$\Delta;\Xi \vdash \mathrm{let}\; a\otimes b \;\mathrm{be}\;x\otimes y\;\mathrm{in} \; c \equiv c[a/x,b/y]:C$}
\DisplayProof \hspace{48pt}
&
\AxiomC{$\Delta;\Xi \vdash \mathrm{let}\; t \;\mathrm{be}\;x\otimes y\;\mathrm{in} \; x\otimes y :A\otimes B$}
\RightLabel{$\otimes$-U}
\UnaryInfC{$\Delta;\Xi \vdash \mathrm{let}\; t \;\mathrm{be}\;x\otimes y\;\mathrm{in} \; x\otimes y\equiv t:A\otimes B$}
\DisplayProof
\\
&\\
&\\
\AxiomC{$\Delta;\cdot \vdash A\;\mathrm{type}$}
\AxiomC{$\Delta;\cdot \vdash B\;\mathrm{type}$}
\RightLabel{$\multimap$-F}
\BinaryInfC{$\Delta;\cdot \vdash A\multimap B\;\mathrm{type}$}
\DisplayProof
&\\
&\\
\AxiomC{$\Delta;\Xi,x:A\vdash b:B$}
\RightLabel{$\multimap$-I}
\UnaryInfC{$\Delta;\Xi\vdash \lambda_{x:A}b:A\multimap B$}
\DisplayProof
&
\AxiomC{$\Delta;\Xi\vdash f:A\multimap B$}
\AxiomC{$\Delta;\Xi'\vdash a:A$}
\RightLabel{$\multimap$-E}
\BinaryInfC{$\Delta;\Xi,\Xi'\vdash f(a):B$}
\DisplayProof
\\
&\\
\AxiomC{$\Delta;\Xi \vdash (\lambda_{x:A}b)(a):B$}
\RightLabel{$\multimap$-C}
\UnaryInfC{$\Delta;\Xi\vdash (\lambda_{x:A}b)(a)\equiv b[a/x]:B$}
\DisplayProof
&
\AxiomC{$\Delta;\Xi \vdash \lambda_{x:A}fx:A\multimap B$}
\RightLabel{$\multimap$-U}
\UnaryInfC{$\Delta;\Xi\vdash \lambda_{x:A}fx\equiv f:A\multimap B$}
\DisplayProof
\end{tabular}
\\
\\
\\
\begin{tabular*}{\textwidth}{lcr}
\AxiomC{}
\RightLabel{$\top$-F}
\UnaryInfC{$\Delta;\cdot\vdash \top\;\mathrm{type}$}
\DisplayProof
\hspace{99pt}
&
\AxiomC{$\vdash \Delta;\Xi\;\mathrm{ctxt}$}
\RightLabel{$\top$-I}
\UnaryInfC{$\Delta;\Xi\vdash \langle\rangle:\top$}
\DisplayProof
\hspace{99pt}
&
\AxiomC{$\Delta;\Xi\vdash t:\top$}
\RightLabel{$\top$-U}
\UnaryInfC{$\Delta;\Xi\vdash t\equiv\langle\rangle:\top$}
\DisplayProof
\end{tabular*}
\\
\\
\\
\begin{tabular}{lr}
\AxiomC{$\Delta;\cdot \vdash A\;\mathrm{type}$}
\AxiomC{$\Delta;\cdot \vdash B\;\mathrm{type}$}
\RightLabel{$\&$-F}
\BinaryInfC{$\Delta;\cdot \vdash A\& B\;\mathrm{type}$}
\DisplayProof
\hspace{154pt}
&
\AxiomC{$\Delta;\Xi\vdash a:A$}
\AxiomC{$\Delta;\Xi\vdash b:B$}
\RightLabel{$\&$-I}
\BinaryInfC{$\Delta;\Xi\vdash \langle a, b\rangle:A\& B$}
\DisplayProof
\\
&\\
\AxiomC{$\Delta;\Xi\vdash t:A\& B$}
\RightLabel{$\&$-E1}
\UnaryInfC{$\Delta;\Xi\vdash \mathrm{fst}(t):A$}
\DisplayProof
&
\AxiomC{$\Delta;\Xi\vdash t:A\& B$}
\RightLabel{$\&$-E2}
\UnaryInfC{$\Delta;\Xi\vdash \mathrm{snd}(t):B$}
\DisplayProof
\\
&\\
\AxiomC{$\Delta;\Xi \vdash \mathrm{fst}(\langle a,b\rangle ):A$}
\RightLabel{$\&$-C1}
\UnaryInfC{$\Delta;\Xi \vdash \mathrm{fst}(\langle a,b\rangle )\equiv a :A$}
\DisplayProof 
&
\AxiomC{$\Delta;\Xi \vdash \mathrm{snd}(\langle a,b\rangle):B$}
\RightLabel{$\&$-C2}
\UnaryInfC{$\Delta;\Xi \vdash \mathrm{snd}(\langle a,b\rangle)\equiv b:B $}
\DisplayProof \\
&\\
&\\
\AxiomC{$\Delta;\Xi\vdash \langle\mathrm{fst}(t),\mathrm{snd}(t) \rangle:A\& B$}
\RightLabel{$\&$-U}
\UnaryInfC{$\Delta;\Xi\vdash \langle\mathrm{fst}(t),\mathrm{snd}(t) \rangle\equiv t:A\& B$}
\DisplayProof
&\\
&\\
\end{tabular}
\begin{tabular}{lcr}
\AxiomC{}
\RightLabel{$0$-F}
\UnaryInfC{$\Delta ; \cdot\vdash 0\;\mathrm{type}$}
\DisplayProof
\hspace{71pt}
&
\AxiomC{$\Delta;\Xi\vdash t:0$}
\RightLabel{$0$-E}
\UnaryInfC{$\Delta;\Xi,\Xi'\vdash \mathrm{false}(t) :B$}
\DisplayProof
\hspace{71pt}
&
\AxiomC{$\Delta;\Xi\vdash t:0$}
\RightLabel{$0$-U}
\UnaryInfC{$\Delta;\Xi\vdash \mathrm{false}(t)\equiv t :0$}
\DisplayProof
\end{tabular}
\\
\\
\\
\begin{tabular}{lr}
\AxiomC{$\Delta;\cdot \vdash A\;\mathrm{type}$}
\AxiomC{$\Delta;\cdot \vdash B\;\mathrm{type}$}
\RightLabel{$\oplus$-F}
\BinaryInfC{$\Delta;\cdot \vdash A\oplus B\;\mathrm{type}$}
\DisplayProof
& \\
& \\
\AxiomC{$\Delta ;\Xi\vdash a: A$}
\RightLabel{$\oplus$-I1}
\UnaryInfC{$\Delta;\Xi\vdash \mathrm{inl}(a): A\oplus B$}
\DisplayProof
&
\hspace{64pt}
\AxiomC{$\Delta ;\Xi\vdash b: B$}
\RightLabel{$\oplus$-I2}
\UnaryInfC{$\Delta;\Xi\vdash \mathrm{inr}(b): A\oplus B$}
\DisplayProof
\\
&\\
&\\
\AxiomC{$\Delta ;\Xi,x:A\vdash c: C$}
\AxiomC{$\Delta ;\Xi,y:B\vdash d: C$}
\AxiomC{$\Delta ;\Xi'\vdash t:A\oplus B$}
\RightLabel{$\oplus$-E}
\TrinaryInfC{$\Delta;\Xi,\Xi'\vdash \mathrm{case}\; t\;\mathrm{of}\;\mathrm{inl}(x)\rightarrow c\;||\;\mathrm{inr}(y)\rightarrow d :C$}
\DisplayProof
&\\
&\\
\AxiomC{$\Delta;\Xi,\Xi'\vdash \mathrm{case}\; \mathrm{inl}(a)\;\mathrm{of}\;\mathrm{inl}(x)\rightarrow c\;||\;\mathrm{inr}(y)\rightarrow d :C$}
\RightLabel{$\oplus$-C1}
\UnaryInfC{$\Delta;\Xi,\Xi'\vdash \mathrm{case}\; \mathrm{inl}(a)\;\mathrm{of}\;\mathrm{inl}(x)\rightarrow c\;||\;\mathrm{inr}(y)\rightarrow d \equiv c[a/x]:C$}
\DisplayProof
&\\
&\\
\AxiomC{$\Delta;\Xi,\Xi'\vdash \mathrm{case}\; \mathrm{inr}(b)\;\mathrm{of}\;\mathrm{inl}(x)\rightarrow c\;||\;\mathrm{inr}(y)\rightarrow d :C$}
\RightLabel{$\oplus$-C2}
\UnaryInfC{$\Delta;\Xi,\Xi'\vdash \mathrm{case}\; \mathrm{inr}(b)\;\mathrm{of}\;\mathrm{inl}(x)\rightarrow c\;||\;\mathrm{inr}(y)\rightarrow d \equiv d[b/y]:C$}
\DisplayProof &\\
&\\
\AxiomC{$\Delta;\Xi\vdash \mathrm{case}\; t\;\mathrm{of}\;\mathrm{inl}(x)\rightarrow \mathrm{inl}(x)\;||\;\mathrm{inr}(y)\rightarrow \mathrm{inr}(y):A\oplus B $}
\RightLabel{$\oplus$-U}
\UnaryInfC{$\Delta;\Xi\vdash \mathrm{case}\; t\;\mathrm{of}\;\mathrm{inl}(x)\rightarrow \mathrm{inl}(x)\;||\;\mathrm{inr}(y)\rightarrow \mathrm{inr}(y)\equiv t:A\oplus B$}
\DisplayProof
\end{tabular}}}
\end{figure}
\clearpage
\begin{figure}
\centering
\fbox{\parbox{\textwidth}{\footnotesize
\quad\\
\begin{tabular}{lr}
\AxiomC{$\Delta;\cdot\vdash A\;\mathrm{type}$}
\RightLabel{$!$-F}
\UnaryInfC{$\Delta;\cdot \vdash !A\;\mathrm{type}$}
\DisplayProof
& \\
& \\
\AxiomC{$\Delta;\cdot \vdash a:A$}
\RightLabel{$!$-I}
\UnaryInfC{$\Delta;\cdot\vdash !a:!A$}
\DisplayProof
&\hspace{128pt}
\AxiomC{$\Delta;\Xi\vdash t:!A$}
\AxiomC{$\Delta,x:A;\Xi'\vdash b:B$}
\RightLabel{$!$-E}
\BinaryInfC{$\Delta;\Xi,\Xi'\vdash \mathrm{let}\; t\;\mathrm{be}\;!x\; \mathrm{in}\; b:B$}
\DisplayProof
\\
&\\
&\\
\AxiomC{$\Delta;\Xi\vdash\mathrm{let}\; !a\;\mathrm{be}\;!x\; \mathrm{in}\; b:B$}
\RightLabel{$!$-C}
\UnaryInfC{$\Delta;\Xi\vdash \mathrm{let}\; !a\;\mathrm{be}\;!x\; \mathrm{in}\; b\equiv b[{a}/x]:B$}
\DisplayProof &
\AxiomC{$\Delta;\Xi\vdash \mathrm{let}\;t\;\mathrm{be}\;!x\; \mathrm{in}\; !x:!A$}
\RightLabel{$!$-U}
\UnaryInfC{$\Delta;\Xi\vdash \mathrm{let}\;t\;\mathrm{be}\;!x\; \mathrm{in}\; !x \equiv t:!A$}
\DisplayProof 
\end{tabular}
\normalsize
}}
\caption{Rules for the usual linear type formers in each context: $I$-, $\otimes$-, $\multimap$-, $\top$-, $\&$-, $0$-, $\oplus$-, and $!$-types.}
\end{figure}
Finally, we add rules for all possible commuting conversions, which from a syntactic point of view restore the subformula property and from a semantic point of view say that our rules are natural transformations (between hom-functors), which simplifies the categorical semantics significantly. We represent these schematically, following \cite{barber1996dual}. That is, if $C[-]$ is a linear program context (rather than a typing context), then we impose the following equations (abusing notation and dealing with all the $\lbi{}{}{}$-constructors in one go).
\begin{figure}[h]
\centering
\fbox{\parbox{\textwidth}{\footnotesize
\quad\\
\AxiomC{$\Delta ;\Xi \vdash C[\lbi{a}{b}{c}]:D$}
\UnaryInfC{$\Delta ;\Xi \vdash C[\lbi{a}{b}{c}]\equiv \lbi{a}{b}{C[c]}:D$}
\DisplayProof
\quad\\
\\
if $C[-]$ does not bind any free variables in $a$ or $b$;
\quad \\
\\
\AxiomC{$\Delta;\Xi\vdash C[\mathrm{false}(t)]:D$}
\UnaryInfC{$\Delta;\Xi\vdash C[\mathrm{false}(t)]\equiv\mathrm{false}(t):D$}
\DisplayProof
\quad\\
\\
if $C[-]$ does not bind any free variables in $t$;
\quad \\
\\
\AxiomC{$\Delta;\Xi\vdash C[\mathrm{case}\;t\;\mathrm{of}\;\mathrm{inl}(x)\rightarrow c\; ||\; \mathrm{inr}(y)\rightarrow d] :D$}
\UnaryInfC{$\Delta;\Xi\vdash C[\mathrm{case}\;t\;\mathrm{of}\;\mathrm{inl}(x)\rightarrow c\; ||\; \mathrm{inr}(y)\rightarrow d]\equiv \mathrm{case}\;t\;\mathrm{of}\;\mathrm{inl}(x)\rightarrow C[c]\; ||\; \mathrm{inr}(y)\rightarrow C[d] :D$}
\DisplayProof
\quad\\
\\
if $C[-]$ does not bind any free variables in $t$ or $x$ or $y$.
\normalsize
}}
\caption{Commuting conversions.}
\end{figure}

\begin{remark}Note that all type formers that are defined context-wise ($I$, $\otimes$, $\multimap$, $\top$, $\&$, $0$, $\oplus$, and $!$) are automatically preserved under the substitutions from Int-Ty-Subst (up to canonical isomorphism\footnote{By an isomorphism of types $\Delta;\cdot\vdash A\;\mathrm{type}$ and $\Delta;\cdot\vdash B\;\mathrm{type}$ in context $\Delta$, we mean here a pair of terms $\Delta;x:A\vdash f:B$ and $\Delta;y:B\vdash g:A$ together with a pair of judgemental equalities $\Delta;x:A\vdash g[f/y]\equiv x:A$ and $\Delta;y:B\vdash f[g/x]\equiv y:B$.}), in the sense that $F(A_1,\ldots, A_n)[{a}/x]$ is isomorphic to $F(A_1[{a}/x],\ldots,A_n[{a}/x])$ for an $n$-ary type former $F$. Similarly, for $T=\Sigma$ or $\Pi$, we have that $(T_{!y:!B}C)[{a}/x]$ is isomorphic to $T_{!y:!B[{a}/x]}C[{a}/x]$ and $(Id_{!B}(b,b'))[a/x]$ is isomorphic to $Id_{!B[a/x]}(b[a/x],b'[a/x])$. (This gives us Beck-Chevalley conditions in the categorical semantics.) These are the remaining naturality conditions for the rules.\end{remark}

\begin{remark}
Note that the usual formulation of universes for DTT transfers naturally to ILDTT, giving us a notion of universes for linear types. This allows us to write rules for forming types as rules for forming terms, as usual. We do not take this approach; we define the various type formers in the setting without universes, as this will give a cleaner categorical semantics.
\end{remark}
\subsubsection*{Some Basic Results} As the focus of this paper is the syntax-semantics correspondence, we will only briefly state a few syntactic results. For some standard metatheoretic properties of the $\multimap,\Pi,\top,\&$-fragment of our syntax, we refer the reader to \cite{cervesato1996linear}. Standard techniques and minor adaptations of the system should suffice to extend the results to all of ILDTT.
\begin{theorem}[Consistency] ILDTT with all its type formers is consistent.\end{theorem}
\begin{proof}We discuss a class of models in section \ref{sec:dismod}.\end{proof}

To give the reader some intuition for these linear $\Pi$- and $\Sigma$-types, we suggest the following two interpretations.

\begin{theorem}[$\Pi$ and $\Sigma$ as Dependent $!(-)\multimap(-)$ and $!(-)\otimes(-)$] Suppose we have $!$-types. Let $\Delta,x:A;\cdot \vdash B\;\mathrm{type}$, where $x$ does not occur freely in $B$. Then, for the purposes of the type theory,
\begin{enumerate} 
\item $\Pi_{!x:{!A}}B$ is isomorphic to $!A\multimap B$, if we have $\Pi$-types and $\multimap$-types;
\item $\Sigma_{!x:{!A}}B$ is isomorphic to $!A\otimes B$, if we have $\Sigma$-types and $\otimes$-types.
\end{enumerate}
\end{theorem}

\begin{proof}\begin{enumerate}
\item We will construct terms
$$\Delta;y:\Pi_{!x:{!A}}B\vdash f:!A\multimap B\txt{and} \Delta;y':!A\multimap B\vdash g:\Pi_{!x:!{A}}B
$$
such that
$$\Delta;y:\Pi_{!x:{!A}}B\vdash g[f/y']\equiv y:\Pi_{!x:{!A}}B\txt{and}\Delta;y':!A\multimap B\vdash f[g/y]\equiv y':!A\multimap B.$$
First, we construct $f$.\\
\\
\AxiomC{}
\RightLabel{Int-Var}
\UnaryInfC{$\Delta,x:A;\cdot \vdash x:A$}
\AxiomC{}
\RightLabel{Lin-Var}
\UnaryInfC{$\Delta,x:A;y:\Pi_{!x:{!A}}B\vdash y:\Pi_{!x:{!A}}B$}
\RightLabel{$\Pi$-E}
\BinaryInfC{$\Delta,x:A;y:\Pi_{!x:{!A}}B\vdash y(!x):B$}
\AxiomC{}
\RightLabel{Lin-Var}
\UnaryInfC{$\Delta;x':!A\vdash x':!A$}
\RightLabel{$!$-E}
\BinaryInfC{$\Delta;y:\Pi_{!x:{!A}}B,x':!A\vdash \mathrm{let}\; x'\;\mathrm{be}\;!x\;\mathrm{in}\;y(!x) :B$}
\RightLabel{$\multimap$-I}
\UnaryInfC{$\Delta;y:\Pi_{!x:!{A}}B\vdash f:!A\multimap B$}
\DisplayProof
\quad\\
\\
Then, we construct $g$.\\
\\
\AxiomC{}
\RightLabel{Int-Var}
\UnaryInfC{$\Delta,x:A;\cdot \vdash x:A$}
\RightLabel{$!$-I}
\UnaryInfC{$\Delta,x:A;\cdot \vdash !x:!A$}
\AxiomC{}
\RightLabel{Lin-Var}
\UnaryInfC{$\Delta,x:A;y':!A\multimap B\vdash y':!A\multimap B$}
\RightLabel{$\multimap$-E}
\BinaryInfC{$\Delta,x:A;y':!A\multimap B\vdash y'(!x):B$}
\RightLabel{$\Pi$-I}
\UnaryInfC{$\Delta;y':!A\multimap B\vdash g:\Pi_{!x:{!A}}B$}
\DisplayProof
\normalsize
\quad\\
\\
It is easily verified that $\multimap$-C, $!$-C, and $\Pi$-U imply the first judgemental equality:\\
 $g[f/y']\equiv \lambda_{!x:{!A}}(\lambda_{x':!A}\mathrm{let}\;x'\;\mathrm{be}\;!x\;\mathrm{in}\;y(!x))(!x)\equiv \lambda_{!x:{!A}}\mathrm{let}\;!x\;\mathrm{be}\;!x\;\mathrm{in}\;y(!x)\equiv \lambda_{!x:{!A}} y(!x)\equiv y$.\\
\\
Similarly, $\Pi$-C, commuting conversions, $!$-U, and $\multimap$-U imply the second judgemental equality:\\ $f[g/y]\equiv \lambda_{x':!A}\mathrm{let}\;x'\;\mathrm{be}\;!x\;\mathrm{in}\;(\lambda_{!x:{!A}}y'(!x))(!x)\equiv \lambda_{x':!A}\mathrm{let}\;x'\;\mathrm{be}\;!x\;\mathrm{in}\;y'(!x)\equiv \lambda_{x':!A}y'(\mathrm{let}\;x'\;\mathrm{be}\;!x\;\mathrm{in}\;!x)\\ \equiv \lambda_{x':!A}y'(x')\equiv y'$.

\item We will construct terms
$$\Delta;y:\Sigma_{!x:{!A}}B\vdash f:!A\otimes B\txt{and} \Delta;y':!A\otimes B\vdash g:\Sigma_{!x:{!A}}B
$$
such that
$$\Delta;y:\Sigma_{!x:{!A}}B\vdash g[f/y']\equiv y:\Sigma_{!x:{!A}}B\txt{and}\Delta;y':!A\otimes B\vdash f[g/y]\equiv y':!A\otimes B.$$
First, we construct $f$.\\
\\
\scriptsize
\AxiomC{}
\RightLabel{Lin-Var}
\UnaryInfC{$\Delta;y:\Sigma_{!x:{!A}}B\vdash y:\Sigma_{!x:!{A}}B$}
\AxiomC{}
\RightLabel{Lin-Var}
\UnaryInfC{$\Delta;x':!A\vdash x':!A$}
\AxiomC{}
\RightLabel{Lin-Var}
\UnaryInfC{$\Delta;z:B\vdash z:B$}
\RightLabel{$\otimes$-I}
\BinaryInfC{$\Delta;x':!A,z:B\vdash x'\otimes z:!A\otimes B$}
\RightLabel{Int-Weak}
\UnaryInfC{$\Delta,x:A;x':!A,z:B\vdash x'\otimes z:!A\otimes B$}
\AxiomC{}
\RightLabel{Int-Var}
\UnaryInfC{$\Delta,x:A;\cdot \vdash x:A$}
\RightLabel{$!$-I}
\UnaryInfC{$\Delta,x:A;\cdot \vdash !x:!A$}
\RightLabel{Lin-Subst}
\BinaryInfC{$\Delta,x:A;z:B\vdash !x\otimes z:!A\otimes B$}
\RightLabel{$\Sigma$-E}
\BinaryInfC{$\Delta;y:\Sigma_{!x:{!A}}B\vdash f:!A\otimes B$}
\DisplayProof
\normalsize
\quad\\
\\
Then, we construct $g$.\\
\\
\tiny
\AxiomC{}
\RightLabel{Lin-Var}
\UnaryInfC{$\Delta;y':!A\otimes B\vdash y':!A\otimes B$}
\AxiomC{}
\RightLabel{Int-Var}
\UnaryInfC{$\Delta,x:A;\cdot\vdash x:A$}
\AxiomC{}
\RightLabel{Lin-Var}
\UnaryInfC{$\Delta,x:A;y:B\vdash y:B$}
\RightLabel{$\Sigma$-I}
\BinaryInfC{$\Delta,x:A;y:B\vdash  !{x} \otimes y :\Sigma_{!x:{!A}}B$}
\AxiomC{}
\RightLabel{Lin-Var}
\UnaryInfC{$\Delta;x':!A\vdash x':!A$}
\RightLabel{$!$-E}
\BinaryInfC{$\Delta;x':!A,y: B\vdash \mathrm{let}\;x'\;\mathrm{be}\;!x\;\mathrm{in}\; ! {x} \otimes y:\Sigma_{!x:{!A}}B$}
\RightLabel{$\otimes$-E}
\BinaryInfC{$\Delta;y':!A\otimes B\vdash g:\Sigma_{!x:{!A}}B$}
\DisplayProof
\normalsize
\quad\\
\\
Here, the first judgemental equality follows from commuting conversions, $\otimes$-C, $!$-C, and $\Sigma$-U:\\
$g[f/y']\equiv\lbi{(\lbi{y}{!{x} \otimes z}{!x\otimes z})}{x'\otimes y}{(\lbi{x'}{!x}{! {x} \otimes y})}\equiv \lbi{y}{!x \otimes z}{\lbi{!x\otimes z}{x'\otimes y}{(\lbi{x'}{!x}{ !{x} \otimes y})}}\equiv  \lbi{y}{!{x} \otimes z}{(\lbi{x'}{!x}{ !{x} \otimes y})[!x/x'][z/y]} \equiv \lbi{y}{!{x} \otimes z}{(\lbi{!x}{!x}{ !{x} \otimes z})} 
\equiv \lbi{y}{!{x} \otimes z}{ !{x} \otimes z} \equiv y$.

The second judgemental equality follows from commuting conversions, $\Sigma$-C, $!$-U, and $\otimes$-U:\\
$f[g/y]\equiv \lbi{(\lbi{y'}{x'\otimes y}{(\lbi{x'}{!x}{ !{x} \otimes y})})}{ !{x} \otimes z}{!x\otimes z}\equiv \lbi{y'}{x'\otimes y}{\lbi{x'}{!x}{\lbi{ !{x} \otimes y}{ !{x} \otimes z}{!x\otimes z}}}\equiv \lbi{y'}{x'\otimes y}{\lbi{x'}{!x}{(!x\otimes y)}}\equiv \lbi{y'}{x'\otimes y}{(\lbi{x'}{!x}{!x})\otimes y}\equiv \lbi{y'}{x'\otimes y}{x'\otimes y}\equiv y'$.
\end{enumerate}
\end{proof}

In particular, we have the following stronger version of a special case.
\begin{theorem}[$!$ as $\Sigma I$]\label{thm:!fromsigma}
Suppose we have $\Sigma$- and $I$-types. Let $\Delta;\cdot \vdash A\;\mathrm{type}$. Then, $\Sigma_{!x:{!A}}I$ satisfies the rules for $!A$. Conversely, if we have $!$- and $I$-types, then $!A$ satisfies the rules for $\Sigma_{!x:{!A}}I$. 
\end{theorem}
\begin{proof}We obtain the $!$-I rule as follows.\\
\\
\AxiomC{$\Delta;\cdot\vdash a:A$}
\AxiomC{}
\RightLabel{$I$-I}
\UnaryInfC{$\Delta,x:A; \cdot \vdash *:I$}
\RightLabel{$\Sigma$-I}
\BinaryInfC{$\Delta;\cdot\vdash  !{a} \otimes *:\Sigma_{!x:{!A}}I$}
\DisplayProof
\quad\\
\\
We obtain the $!$-E rule as follows.\\
\\
\AxiomC{$\Delta;\Xi\vdash t:\Sigma_{!x:{!A}}I$}
\AxiomC{$\Delta,x:A;\Xi'\vdash c:C$}
\AxiomC{}
\RightLabel{Lin-Var}
\UnaryInfC{$\Delta;y:I\vdash y:I$}
\RightLabel{$I$-E}
\BinaryInfC{$\Delta,x:A;\Xi',y:I\vdash \mathrm{let}\; y\;\mathrm{be}\; *\;\mathrm{in}\;c:C$}
\RightLabel{$\Sigma$-E}
\BinaryInfC{$\Delta;\Xi,\Xi'\vdash \mathrm{let}\; t\;\mathrm{be}\; !{x} \otimes y\;\mathrm{in}\;\mathrm{let}\; y\;\mathrm{be}\; *\;\mathrm{in}\;c:C$.}
\DisplayProof
\quad\\
\\
It is easily seen that $\Sigma$-C and $I$-C imply $!$-C ($\mathrm{let}\;!{a} \otimes *\;\mathrm{be}\;{!x} \otimes y\;\mathrm{in}\; \mathrm{let}\; y\;\mathrm{be}\; *\;\mathrm{in}\;c\equiv(\mathrm{let}\; y\;\mathrm{be}\; *\;\mathrm{in}\;c)[{a}/x][*/y]\equiv \mathrm{let}\; *\;\mathrm{be}\; *\;\mathrm{in}\;c[{a}/x]  \equiv c[{a}/x]$) and that $I$-U and $\Sigma$-U (and commuting conversions) imply $!$-U ($\mathrm{let}\;t\;\mathrm{be}\;!{x} \otimes y\;\mathrm{in}\;\mathrm{let}\; y\;\mathrm{be}\; *\;\mathrm{in}\;!{x} \otimes *\equiv \mathrm{let}\;t\;\mathrm{be}\;!{x} \otimes y\;\mathrm{in}\;!{x} \otimes \mathrm{let}\; y\;\mathrm{be}\; *\;\mathrm{in}\;* \equiv \mathrm{let}\;t\;\mathrm{be}\;!{x} \otimes y\;\mathrm{in}\;!{x} \otimes y\equiv t$).\\
\\
The converse statement follows by a similar argument, noting that $I[{a}/x]$ is isomorphic to $I$.
\end{proof}

A second interpretation is that $\Pi$ and $\Sigma$ generalise $\&$ and $\oplus$. Indeed, the idea is that these (or their infinitary equivalents) are what they reduce to when taken over discrete types. The subtlety in this result is the definition of a discrete type. The same phenomenon is observed in a different context in section \ref{sec:dismod}.

For our purposes, a discrete type is a strong sum of $I$ (a sum with a dependent -E-rule). For simplicity, let us limit ourselves to the binary case. For us, the discrete type with two elements will be $2=I\oplus I$, where $\oplus$ has a strong/dependent -E-rule (note that this is not our $\oplus$-E). Explicitly, $2$ is a type with the following rules:
\begin{figure}[h]
\quad\\
\fbox{\parbox{\textwidth}{\footnotesize
\quad\\
\begin{tabular}{lcr}
\AxiomC{}
\RightLabel{$2$-F}
\UnaryInfC{$\Delta;\cdot\vdash 2\;\mathrm{type}$}
\DisplayProof
\hspace{110pt}
&
\AxiomC{}
\RightLabel{$2$-I1}
\UnaryInfC{$\Delta;\cdot\vdash \mathrm{tt}:2$}
\DisplayProof
\hspace{110pt}
&
\AxiomC{}
\RightLabel{$2$-I2}
\UnaryInfC{$\Delta;\cdot\vdash \mathrm{ff}:2$}
\DisplayProof
\end{tabular}
\\
\\
\\
\begin{tabular}{c}
\AxiomC{$\Delta,x:2;\cdot \vdash A\;\mathrm{type}$}
\AxiomC{$\Delta;\cdot\vdash t:2$}
\AxiomC{$\Delta;\Xi\vdash a_{\mathrm{tt}}:A[{\mathrm{tt}}/x]$}
\AxiomC{$\Delta;\Xi\vdash a_{\mathrm{ff}}:A[{\mathrm{ff}}/x]$}
\RightLabel{$2$-E}
\QuaternaryInfC{$\Delta;\Xi\vdash \mathrm{if}\; t\;\mathrm{then}\; a_{\mathrm{tt}}\;\mathrm{else}\;a_{\mathrm{ff}}:A[{t}/x]$}
\DisplayProof
\normalsize
\end{tabular}
\\
\\
\\
\begin{tabular}{lr}
\AxiomC{$\Delta;\Xi\vdash \mathrm{if}\; \mathrm{tt}\;\mathrm{then}\; a_{\mathrm{tt}}\;\mathrm{else}\;a_{\mathrm{ff}}:A[\mathrm{tt}/x]$}
\RightLabel{$2$-C1}
\UnaryInfC{$\Delta;\Xi\vdash \mathrm{if}\; \mathrm{tt}\;\mathrm{then}\; a_{\mathrm{tt}}\;\mathrm{else}\;a_{\mathrm{ff}}\equiv a_{\mathrm{tt}}:A[\mathrm{tt}/x]$}
\DisplayProof
&\hspace{68pt}
\AxiomC{$\Delta;\Xi\vdash \mathrm{if}\; \mathrm{ff}\;\mathrm{then}\; a_{\mathrm{tt}}\;\mathrm{else}\;a_{\mathrm{ff}}:A[\mathrm{ff}/x]$}
\RightLabel{$2$-C2}
\UnaryInfC{$\Delta;\Xi\vdash \mathrm{if}\; \mathrm{ff}\;\mathrm{then}\; a_{\mathrm{tt}}\;\mathrm{else}\;a_{\mathrm{ff}}\equiv a_{\mathrm{ff}}:A[\mathrm{ff}/x]$}
\DisplayProof\\
&\\

\AxiomC{$\Delta;\Xi\vdash \mathrm{if}\; t\;\mathrm{then}\; \mathrm{tt}\;\mathrm{else}\;\mathrm{ff}:2$}
\RightLabel{$2$-U}
\UnaryInfC{$\Delta;\Xi\vdash \mathrm{if}\; t\;\mathrm{then}\; \mathrm{tt}\;\mathrm{else}\;\mathrm{ff}\equiv t:2$}
\DisplayProof
\end{tabular}
}}
\caption{Rules for a discrete type $2$.}
\end{figure}
\begin{theorem}[$\Pi$ and $\Sigma$ as Infinitary $\&$ and $\oplus$]\label{thm:pisigmainf} If we have a discrete type $2$ and a type family $\Delta,x: 2;\cdot\vdash A\;\mathrm{type}$, then
\begin{enumerate}
\item $\Pi_{!x:{!2}}A$ satisfies the rules for $A[{\mathrm{tt}}/x]\& A[{\mathrm{ff}}/x]$;
\item $\Sigma_{!x:{!2}}A$ satisfies the rules for $A[{\mathrm{tt}}/x]\oplus A[{\mathrm{ff}}/x]$.
\end{enumerate}
\end{theorem}
\begin{proof}\begin{enumerate}
\item We obtain $\&$-I as follows.

\footnotesize
\AxiomC{$\Delta,x:2;\Xi\vdash a:A[{\mathrm{tt}}/x]$}
\AxiomC{$\Delta,x:2;\Xi\vdash b:A[{\mathrm{ff}}/x]$}
\AxiomC{}
\RightLabel{Int-Var}
\UnaryInfC{$\Delta,x:2;\cdot\vdash x:2$}
\AxiomC{}
\RightLabel{Assumption}
\UnaryInfC{$\Delta,x:2;\cdot \vdash A\;\mathrm{type}$}
\RightLabel{2-E-dep}
\QuaternaryInfC{$\Delta,x:2;\Xi\vdash \mathrm{if}\; x\;\mathrm{then}\; a\;\mathrm{else}\; b:A$}
\RightLabel{$\Pi$-I}
\UnaryInfC{$\Delta;\Xi\vdash \lambda_{!x:{!2}}\mathrm{if}\; x\;\mathrm{then}\; a\;\mathrm{else}\; b:\Pi_{!x:{!2}}A$}
\DisplayProof
\normalsize

Moreover, we obtain $\&$-E1 as follows (similarly, we obtain $\&$-E2).

\AxiomC{$\Delta;\Xi\vdash t:\Pi_{!x:!2}A$}
\AxiomC{}
\RightLabel{2-I1}
\UnaryInfC{$\Delta;\cdot \vdash \mathrm{tt}:2$}
\RightLabel{$\Pi$-E}
\BinaryInfC{$\Delta;\Xi\vdash t(!\mathrm{tt}):A[{\mathrm{tt}}/x]$}
\DisplayProof

The $\&$-C-rules follow from $\Pi$-C and $2$-C, e.g.\\ $\mathrm{fst}\langle a,b\rangle :\equiv (\lambda_{!x:{!2}}\mathrm{if}\; x\;\mathrm{then}\; a\;\mathrm{else}\; b)(!\mathrm{tt})
\equiv \mathrm{if}\; \mathrm{tt}\;\mathrm{then}\; a\;\mathrm{else}\; b\equiv a$.

The $\&$-U-rules follow from $\Pi$-U and $2$-U (and commuting conversions):\\ $\langle \mathrm{fst}(t),\mathrm{snd}(t)\rangle:\equiv \lambda_{!x:!{2}}\mathrm{if}\; x\;\mathrm{then}\; t(!\mathrm{tt})\;\mathrm{else}\; t(!\mathrm{ff})\equiv \lambda_{!x:{!2}}t(!x)\equiv t$.

\item We obtain $\oplus$-I1 as follows (and similarly, we obtain $\oplus$-I2):

\AxiomC{}
\RightLabel{2-I1}
\UnaryInfC{$\Delta;\cdot \vdash \mathrm{tt}:2$}
\AxiomC{$\Delta;\Xi\vdash a:A[{\mathrm{tt}}/x]$}
\RightLabel{$\Sigma$-I}
\BinaryInfC{$\Delta;\Xi\vdash  {!\mathrm{tt}} \otimes a:\Sigma_{!x:{!2}}A$}
\DisplayProof

Moreover, we obtain $\oplus$-E as follows.\\
\\
\tiny

\AxiomC{$\Delta;\Xi'\vdash t:\Sigma_{!x:{!2}}A$}
\AxiomC{$\Delta;\Xi,z:A[{\mathrm{tt}}/x]\vdash c:C$}
\RightLabel{}
\UnaryInfC{$\Delta,x:2;\Xi,y:A\vdash c[y/z]:C$}
\AxiomC{$\Delta;\Xi,w:A[{\mathrm{ff}}/x]\vdash d:C$}
\RightLabel{}
\UnaryInfC{$\Delta,x:2;\Xi,y:A\vdash d[y/w]:C$}
\AxiomC{}
\RightLabel{}
\UnaryInfC{$\Delta,x:2;\cdot\vdash x:2$}
\AxiomC{}
\RightLabel{Assump}
\UnaryInfC{$\Delta,x:2;\cdot\vdash A\;\mathrm{type}$}
\RightLabel{2-E-dep}
\QuaternaryInfC{$\Delta,x:2;\Xi,y:A\vdash \mathrm{if}\;x\;\mathrm{then}\;c[y/z]\;\mathrm{else}\;d[y/w]:C$}
\RightLabel{$\Sigma$-E}
\BinaryInfC{$\Delta;\Xi,\Xi'\vdash \mathrm{let}\;t \;\mathrm{be}\; !{x} \otimes y \;\mathrm{in}\; \mathrm{if}\;x\;\mathrm{then}\;c[y/z]\;\mathrm{else}\;d[y/w] :C$}
\DisplayProof
\normalsize

The $\oplus$-C-rules follow from $\Sigma$-C and 2-C, e.g.
\begin{align*} \mathrm{case}\; \mathrm{inl}(a)\;\mathrm{of}\;\mathrm{inl}(z)\rightarrow c||\;\mathrm{inr}(w)\rightarrow d: & \equiv \mathrm{let}\; {!\mathrm{tt}} \otimes a \;\mathrm{be}\; !{x} \otimes y \;\mathrm{in}\; \mathrm{if}\;x\;\mathrm{then}\;c[y/z]\;\mathrm{else}\;d[y/w]\\
&\equiv c[a/z].
\end{align*}

The $\oplus$-U-rules follow from $\Sigma$-U and 2-U (and commuting conversions):
\begin{align*} \mathrm{case}\; t\;\mathrm{of}\;\mathrm{inl}(z)\rightarrow \mathrm{inl}(z)||\;\mathrm{inr}(w)\rightarrow \mathrm{inr}(w): & \equiv \mathrm{let}\;t \;\mathrm{be}\; !{x} \otimes y \;\mathrm{in}\; \mathrm{if}\;x\;\mathrm{then}\;\mathrm{inl}(z)[y/z]\;\mathrm{else}\;\mathrm{inr}(w)[y/w]\\
&\equiv \mathrm{let}\;t \;\mathrm{be}\;! {x} \otimes y \;\mathrm{in}\; \mathrm{if}\;x\;\mathrm{then}\; {!\mathrm{tt}} \otimes z[y/z]\;\mathrm{else}\; {!\mathrm{ff}} \otimes w[y/w]\\
&\equiv \mathrm{let}\;t \;\mathrm{be}\; !{x} \otimes y \;\mathrm{in}\; \mathrm{if}\;x\;\mathrm{then}\; {!\mathrm{tt}} \otimes y\;\mathrm{else}\; {!\mathrm{ff}} \otimes y\\
&\equiv \mathrm{let}\;t \;\mathrm{be}\; !{x} \otimes y \;\mathrm{in}\; {!(\mathrm{if}\;x\;\mathrm{then}\; \mathrm{tt}\;\mathrm{else}\;\mathrm{ff})} \otimes y\\
&\equiv \mathrm{let}\;t \;\mathrm{be}\;! {x} \otimes y \;\mathrm{in}\;  !{x} \otimes y\\
&\equiv t.
\end{align*}

\end{enumerate}
\end{proof}
We see that we can view $\Pi$ and $\Sigma$ as generalisations of $\&$ and $\oplus$, respectively.\clearpage

\section{Semantics of ILDTT}
\label{sec:sem}
The idea behind the categorical semantics we present for the structural core of our syntax (with $I$- and $\otimes$-types) is to take our suggested categorical semantics for the structural core of DTT (with $\top$- and $\wedge$-types) and relax the assumption that its fibres are Cartesian, requiring only that they be (possibly non-Cartesian) symmetric monoidal. This exactly reflects the relation between the conventional semantics of non-dependent intuitionistic and linear type systems. The structure we obtain is that of a strict indexed symmetric monoidal\footnote{As noted above, we can easily obtain a sound and complete semantics for only the structural core, possibly without $I$- and $\otimes$-types, by considering strict indexed symmetric multicategories with comprehension.} category with comprehension.

The $\Sigma$- and $\Pi$-types arise as left and right adjoints of substitution functors along projections in the base category and the $\mathrm{Id}$-types arise as left adjoints to substitution along diagonals, all satisfying Beck-Chevalley (and Frobenius) conditions, as is the case in the semantics for DTT. The $!$-types amount to having a left adjoint to the comprehension (which can be made a functor), giving a linear-non-linear adjunction as in the conventional semantics for linear logic. Finally, additive connectives arise as compatible Cartesian and distributive co-Cartesian structures on the fibres, as would be expected from the semantics of linear logic.

\subsection{Tautological models of ILDTT}
First, we translate the structural core of our syntax to the tautological notion of model. We later prove that this is equivalent to the more intuitive notion of categorical model referred to above.

\begin{definition*}[Tautological model of ILDTT] By a (tautological) model $\widetilde{\mathbb{T}}$ of ILDTT, we mean the following.
\begin{enumerate}
\item We have a set $\mathrm{ICtxt}$, whose elements will be interpreted as (dependent) contexts consisting of intuitionistic types. Then, for all $\Delta\in\mathrm{ICtxt}$, we have a set $\mathrm{LType}(\Delta)$ of linear types and a set $\mathrm{LCtxt}(\Delta)$ of linear contexts (multisets of linear types) in the context $\Delta$. For each $\Delta$, $\Xi\in\mathrm{LCtxt}(\Delta)$, and each $A\in \mathrm{LType}(\Delta)$, we have a set $\mathrm{LTerm}(\Delta,\Xi,A)$ of (linear) terms of $A$ in context $\Delta;\Xi$. On all these sets, judgemental equality is interpreted by equality of elements, taking into account $\alpha$-conversion for terms. (This means that, if we construct a model from the syntax, we divide out judgemental equality on the syntactic objects.) This handles the -Eq rules, the rules expressing that judgemental equality is an equivalence relation, and the rules relating typing and judgemental equality. 

\item C-Emp says that $\mathrm{ICtxt}$ has a distinguished element $\cdot$ and that $\mathrm{LCtxt}(\cdot)$ has a distinguished element $\cdot$.

\item Int-C-Ext says that for all $\Delta\in\mathrm{ICtxt}$ and $A\in\mathrm{LType}(\Delta)$, we can form a context $\Delta.{A}\in\mathrm{ICtxt}$, and that we have a function $\mathrm{LCtxt}(\Delta)\ra{}\mathrm{LCtxt}(\Delta.{A})$, introducing fake dependencies of the linear types on ${A}$. Int-Exch says that for two types in the same context, the order in which we append them to a context does not matter.

\item Lin-C-Ext says that for all $\Xi\in\mathrm{LCtxt}(\Delta)$ and $A\in\mathrm{LType}(\Delta)$, we can form a context $\Xi.A\in\mathrm{LCtxt}(\Delta)$. Lin-Exch says that the order in which we do this does not matter.

\item Int-Var says that for a context $\Delta.{A}.\Delta'$, we have a term $\mathrm{der}_{\Delta,A,\Delta'}\in \mathrm{LTerm}(\Delta.{A}.\Delta',\cdot,A)$, which---this is implicitly present in the syntax---acts as a diagonal morphism through the substitution operations of Int-Ty-Subst and Int-Tm-Subst, equating the values of two variables of type $A$.

\item Lin-Var says that for a context $\Delta$ and a type $A\in\mathrm{LType}(\Delta)$, we have a term $\mathrm{id}_A\in\mathrm{LTerm}(\Delta,A,A)$, which---this is implicitly present in the syntax---acts as the identity for the substitution operation of Lin-Tm-Subst.

\item Int-Weak says that we have functions (abusing notation) $\mathrm{LType}(\Delta.\Delta')\ra{\mathrm{weak}_{\Delta,{A},\Delta'}}\mathrm{LType}(\Delta.{A}.\Delta')$ and
$\mathrm{LTerm}(\Delta.\Delta',\Xi,B)\ra{\mathrm{weak}_{\Delta,{A},\Delta'}}\mathrm{LTerm}(\Delta.{A}.\Delta',\Xi,B)$. We think of this as projecting away a variable $x:A$ to introduce a fake dependency.

\item Int-Ty-Subst says that for $B\in\mathrm{LType}(\Delta.{A}.\Delta')$ and $a\in\mathrm{LTerm}(\Delta,\cdot,A)$, we have a context $\Delta.\Delta'[{a}/x]\in\mathrm{ICtxt}$ and a $B[{a}/x]\in \mathrm{LType}(\Delta.\Delta'[{a}/x])$.

\item Int-Tm-Subst says that for $b\in\mathrm{LTerm}(\Delta.{A}.\Delta',\Xi,B)$ and $a\in\mathrm{LTerm}(\Delta,\cdot,A)$, we have $b[{a}/x]\in \mathrm{LTerm}(\Delta.\Delta'[{a}/x],\Xi[{a}/x],B[{a}/x])$.

\item Lin-Tm-Subst says that for $b\in\mathrm{LTerm}(\Delta,\Xi.A,B)$ and $a\in\mathrm{LTerm}(\Delta,\Xi',A)$, we have $b[a/x]\in\mathrm{LTerm}(\Delta,\Xi.\Xi',B)$.
\end{enumerate}
For these last three substitution operations, it is implicit in the syntax that they are associative.

Finally, the remarks in Int-Var and Int-Weak about diagonals and projections in formal terms mean that the morphisms coming from these rules work together to form a generalised comonoid.
\end{definition*}
It is tautological that there is a one-to-one correspondence between theories $\mathbb{T}$ in ILDTT and models $\widetilde{\mathbb{T}}$ of this sort.\\
\\
We now define what it means for a model to support various type formers.
\begin{definition*}[Semantic $I$- and $\otimes$-types] We say a model $\widetilde{\mathbb{T}}$ supports $I$-types if for all $\Delta\in\mathrm{ICtxt}$, we have an $I\in\mathrm{LType}(\Delta)$ and $*\in\mathrm{LTerm}(\Delta,\cdot, I)$ and whenever $t\in\mathrm{LTerm}(\Delta,\Xi',I)$ and $a\in\mathrm{LTerm}(\Delta,\Xi,A)$, we have $\mathrm{let}\; t\;\mathrm{be}\;*\;\mathrm{in}\; a\in\mathrm{LTerm}(\Delta,\Xi.\Xi',A)$, such that $\mathrm{let}\; *\;\mathrm{be}\;*\;\mathrm{in}\;a=a$ and $\mathrm{let}\;t\;\mathrm{be}\;*\;\mathrm{in}\;*=t$.

Similarly, we say it supports $\otimes$-types if for all $A,B\in\mathrm{LType}(\Delta)$, we have an $A\otimes B\in\mathrm{LType}(\Delta)$, for all $a\in \mathrm{LTerm}(\Delta,\Xi,A), b\in\mathrm{LTerm}(\Delta,\Xi',B)$, we have $a\otimes b\in\mathrm{LTerm}(\Delta,\Xi.\Xi',A\otimes B)$, and if $t\in \mathrm{LTerm}(\Delta,\Xi,A\otimes B)$ and $c\in\mathrm{LTerm}(\Delta,\Xi'.A.B,C)$, we have $\mathrm{let}\; t\;\mathrm{be}\; x\otimes y\;\mathrm{in}\;c\in\mathrm{LTerm}(\Delta,\Xi.\Xi',C)$, such that $\mathrm{let}\;a\otimes b\;\mathrm{be}\;x\otimes y\;\mathrm{in}\;c=c$ and $\mathrm{let}\;t\;\mathrm{be}\;x\otimes y\;\mathrm{in}\;x\otimes y=t$.

Note that this defines a function $\mathrm{LCtxt}(\Delta)\ra{\bigotimes} \mathrm{LType}$. The C-rule precisely says that from the point of view of the (terms of the) type theory this map is an injection, while the U-rule says it is a surjection\footnote{The precise statement that we are alluding to here would be that the multicategory of linear contexts is equivalent to the (monoidal) multicategory of linear types. More precisely, $\bigotimes$ is only part of an equivalence of categories rather than an isomorphism, i.e. it is injective on objects up to isomorphism rather than on the nose.}. We conclude that in the presence of $I$- and $\otimes$-types, we can faithfully describe the type theory without mentioning linear contexts, replacing them by the linear type that is their $\otimes$-product.
\end{definition*}
We will henceforth assume that our type theory has $I$- and $\otimes$-types, as this simplifies the categorical semantics\footnote{To be precise, it allows us to give a categorical semantics in terms of monoidal categories rather than multicategories.} and is appropriate for the examples we are interested in.

For the other type formers, one can give a similar, almost tautological, translation from the syntax into a tautological model. We leave this to the reader, and discuss the semantic counterparts of various type formers in the categorical semantics presented next.

\subsection{Categorical Semantics of ILDTT}
\subsubsection*{Strict Indexed Symmetric Monoidal Categories with Comprehension}
We now introduce a notion of categorical model for which soundness and completeness results hold with respect to the syntax of ILDTT in the presence of $I$- and $\otimes$-types\footnote{If we are interested in the case without $I$- and $\otimes$-types, the semantics easily generalises to strict indexed symmetric multicategories with comprehension.}. This notion of model will prove to be particularly useful when thinking about various (extensional) type formers.

\begin{definition}By a \emph{strict indexed symmetric monoidal category with comprehension}, we mean the following data.
\begin{enumerate}
\item A category $\mathcal{C}$ with a terminal object $\cdot$.
\item A strict indexed symmetric monoidal category $\mathcal{L}$ over $\mathcal{C}$, i.e. a contravariant functor $\mathcal{L}$ into the category $\mathrm{SMCat}$ of (small) symmetric monoidal categories and strong monoidal functors $\mathcal{C}^{op}\ra{\mathcal{L}}\mathrm{SMCat}.$
We will also write $-\{f\}:=\mathcal{L}(f)$ for the action of $\mathcal{L}$ on a morphism $f$ of $\mathcal{C}$.
\item A \emph{comprehension schema}, i.e. for each $\Delta\in\mathrm{ob}(\mathcal{C})$ and $A\in\mathrm{ob}(\mathcal{L}(\Delta))$ a representation for the functor $$x\mapsto\mathcal{L}(\mathrm{dom}(x))(I,A\{x\}):(\mathcal{C}/\Delta)^{op}\ra{}\mathrm{Set}.$$ We will write its representing object\footnote{More precisely, $\Delta.MA\ra{\mathbf{p}_{\Delta,MA}}\Delta$ would be a better notation, where we think of $L\dashv M$ as an adjunction inducing $!$, but it would be very verbose.} $\Delta.{A}\ra{\mathbf{p}_{\Delta,{A}}}\Delta\in\mathrm{ob}(\mathcal{C}/\Delta)$ and universal element $\mathbf{v}_{\Delta,{A}}\in\mathcal{L}(\Delta.{A})(I,A\{\mathbf{p}_{\Delta,{A}}\})$. We will write $a  \mapsto  \langle f,a\rangle$ for the isomorphism $\mathcal{L}(\Delta')(I,A\{f\}) \cong  \mathcal{C}/\Delta(f,\mathbf{p}_{\Delta,{A}})$.
\end{enumerate}
\end{definition}
Again, the comprehension schema means that the morphisms in our category of contexts $\mathcal{C}$ into a context built by adjoining types arise as lists of closed linear terms. Here, the crucial point is the identification of intuitionistic terms with linear terms without linear assumptions: they can be freely copied and discarded.
\begin{remark}Note that this notion of model reduces to a standard notion of model for intuitionistic dependent type theory when the monoidal structures on the fibre categories are Cartesian: a strict indexed Cartesian monoidal category with comprehension. Indeed, these are easily seen to be equivalent to, for instance, the split comprehension categories of \cite{jacobs1993comprehension} with terminal object and Cartesian products. However, it turns out that we have to impose the extra requirement of fullness on the comprehension category to get an exact match with the syntax. The corresponding condition in our framework is to ask for the comprehension functor to be full and faithful. See \cite{jacobs1993comprehension} for more discussion.
\end{remark}
\begin{theorem}[Comprehension functor]\label{thm:comprfunc} A comprehension schema $(\mathbf{p},\mathbf{v})$ on a strict indexed symmetric monoidal category $(\mathcal{C},\mathcal{L})$ defines a morphism $\mathcal{L}\ra{M}\mathcal{I}$ of indexed categories, which laxly sends the monoidal structure of $\mathcal{L}$ to products in $\mathcal{I}$ (where they exist). Here, $\mathcal{I}$ is the full subindexed\footnote{Here, we use the axiom of choice to make a choice of pullback and make $\mathcal{I}$ into an indexed category (or cloven fibration). Alternatively, we can avoid the axiom of choice and treat it as a more general fibration.} category of $\mathcal{C}/-$ on the objects of the form $\mathbf{p}_{\Delta,{A}}$.\end{theorem}
\begin{proof}First, note that a morphism $M$ of indexed symmetric monoidal categories consists of lax monoidal functors in each context $\Delta\in\mathcal{C}$ such that
\[
\begin{tikzcd}[ampersand replacement=\&, column sep=huge, row sep=large]
\mathcal{L}(\Delta) \arrow[r,"{M_\Delta}"] \arrow[d,"{\mathcal{L}(f)}"'] \arrow[dr, phantom,"{\cong}" description] \& \mathcal{I}(\Delta) \arrow[d,"{\mathcal{I}(f)=\textnormal{``pullback along }f\textnormal{''}}"]\\
\mathcal{L}(\Delta^{\prime}) \arrow[r,"{M_{\Delta^{\prime}}}"'] \& \mathcal{I}(\Delta^{\prime}).
\end{tikzcd}
\]

We define
\[
\begin{tikzcd}[ampersand replacement=\&, column sep=huge]
M_\Delta(A\ra{a}B):=\mathbf{p}_{\Delta,{A}} \arrow[r,"{\langle\mathbf{p}_{\Delta,{A}},a\{\mathbf{p}_{\Delta,{A}}\}\circ \mathbf{v}_{\Delta,{A}}\rangle}"] \& \mathbf{p}_{\Delta,{B}}.
\end{tikzcd}
\]
Functoriality follows from the uniqueness property of $\langle \mathrm{id}_\Delta,a\rangle$.

We define the lax monoidal structure
\[
\begin{tikzcd}[ampersand replacement=\&, column sep=large, row sep=small]
\mathrm{id}_\Delta \arrow[r,"{m_\Delta^I}"] \& M_\Delta(I)=\mathbf{p}_{\Delta,{I}}\\
\mathbf{p}_{\Delta,{A}}\circ \mathbf{p}_{\Delta.{A},{B\{\mathbf{p}_{\Delta,A}\}}}=M_\Delta(A)\times M_\Delta(B) \arrow[r,"{m_\Delta^{A,B}}"'] \& M_\Delta(A\otimes B)=\mathbf{p}_{\Delta,{A\otimes B}},
\end{tikzcd}
\]
where $m_\Delta^I:=\langle \mathrm{id}_\Delta,\mathrm{id}_{I} \rangle$ and $m_\Delta^{A,B}:=\langle \mathbf{p}_{\Delta,{A}}\circ\mathbf{p}_{\Delta.{A}, B\{\mathbf{p}_{\Delta,A}\}},\mathbf{v}_{\Delta,{A}}\{\mathbf{p}_{\Delta.{A}, B\{\mathbf{p}_{\Delta,A}\}}\}\otimes\mathbf{v}_{\Delta.{A},{B\{\mathbf{p}_{\Delta,A}\}}} \rangle$.

Finally, we verify that $\mathcal{I}(f)M_\Delta=M_{\Delta'}\mathcal{L}(f)$. This follows directly from the fact that the following square is a pullback square:
\[
\begin{tikzcd}[ampersand replacement=\&, column sep=large, row sep=large]
\Delta^{\prime}.{A\{f\}} \arrow[r,"{\mathbf{q}_{f,{A}}}"] \arrow[d,"{\mathbf{p}_{\Delta^{\prime},{A\{f\}}}}"'] \& \Delta.{A} \arrow[d,"{\mathbf{p}_{\Delta,{A}}}"]\\
\Delta^{\prime} \arrow[r,"{f}"'] \& \Delta,
\end{tikzcd}
\]
where $\mathbf{q}_{f,{A}}:=\langle f\mathbf{p}_{\Delta',{A\{f\}}},\mathbf{v}_{\Delta',{A\{f\}}}\rangle$. We leave this verification to the reader as an exercise. Alternatively, a proof of this fact in DTT that transfers entirely to our setting can be found in \cite{hofmann1997syntax}.
\end{proof}
\begin{remark}
Note that $\mathcal{I}$ is a display map category (or, less specifically, a full comprehension category). Hence, it is a model of intuitionistic type theory. We will see that, in many ways, we can regard it as the intuitionistic content of $\mathcal{L}$.
\end{remark}
\begin{remark}We will see that this functor gives us a unique candidate for $!$-types: $!:=LM$, where $L\dashv M$. We conclude that, in ILDTT, the $!$-modality is uniquely determined by the indexing. This is worth noting, because, in propositional linear type theory, we might have many different candidates for $!$-types.

Moreover, it explains why we do not demand that $M$ be fully faithful in the case of linear types. Indeed, although we have a map $\mathcal{L}(\Delta)(A,B)\ra{M_\Delta}\mathcal{I}(\Delta)(\mathbf{p}_{\Delta,{A}},\mathbf{p}_{\Delta,{B}})\cong \mathcal{L}(\Delta.{A})(I,B\{\mathbf{p}_{\Delta,{A}}\})$, this is not generally an isomorphism. In fact, in the presence of $!$-types, we will see that the right-hand side is precisely isomorphic to $\mathcal{L}(\Delta)(!A,B)$ and the map is precomposition with dereliction.
\end{remark}

Next, we prove that we have a sound interpretation of ILDTT in such categories.

\begin{theorem}[Soundness] A strict indexed symmetric monoidal category with comprehension $(\mathcal{C},\mathcal{L},\mathbf{p},\mathbf{v})$ defines a model $\widetilde{\mathbb{T}}^{(\mathcal{C},\mathcal{L},\mathbf{p},\mathbf{v})}$ of ILDTT with $I$- and $\otimes$-types.
\end{theorem}
\begin{proof}
We define
\begin{enumerate}
\item $\mathrm{ICtxt}:=\mathrm{ob}(\mathcal{C})$\\
 $\mathrm{LType}(\Delta):=\mathrm{ob}(\mathcal{L}(\Delta))$\\
 $\mathrm{LCtxt}(\Delta):=\mathrm{free-comm.-monoid}(\mathrm{LType}(\Delta))$ (where we will write $0$ and $+$ for the operations)\\
 $\mathrm{LTerm}(\Delta,\Xi,A):=\mathcal{L}(\Delta)(\bigotimes\Xi,A)$\\
\item C-Emp: $\cdot_\mathrm{ICtxt}:=\cdot_{\mathcal{C}}$ and $\cdot_{\mathrm{LCtxt}(\Delta)}:=0_{\mathrm{LCtxt}(\Delta)}$.\\
\item Int-C-Ext: $\Delta._\mathrm{ICtxt}{A}:=\Delta._{\mathcal{C}}{A}$ and $\mathrm{LType}(\Delta)\ra{}\mathrm{LType}(\Delta.{A}):=-\{\mathbf{p}_{\Delta,{A}}\}$ inducing the obvious function $\mathrm{LCtxt}(\Delta)\ra{}\mathrm{LCtxt}(\Delta.{A})$. We have seen how
\[
\begin{tikzcd}[ampersand replacement=\&, column sep=large, row sep=large]
\Delta.A.B\{\mathbf{p}_{\Delta,A}\} \arrow[r,"{\mathbf{q}_{\mathbf{p}_{\Delta,A},B}}"] \arrow[d,"{\mathbf{p}_{\Delta.A,B\{\mathbf{p}_{\Delta,A}\}}}"'] \& \Delta.B \arrow[d,"{\mathbf{p}_{\Delta,B}}"]\\
\Delta.A \arrow[r,"{\mathbf{p}_{\Delta,A}}"'] \& \Delta
\end{tikzcd}
\]
is a product in $\mathcal{C}/\Delta$ which interprets the double context extension $\Delta.A.B$ where $A,B\in \mathrm{LType}(\Delta)$. Being a Cartesian monoidal structure, this is, in particular, symmetric, and so validates Int-Exch.\\
\item Lin-C-Ext: $\Xi.A:=\Xi+ A$\\
\item Int-Var: $\mathrm{der}_{\Delta,A,\Delta'}\in\mathrm{LTerm}(\Delta.{A}.\Delta',\cdot,A)$ is defined as $$\mathbf{v}_{\Delta,{A}}\{\mathbf{p}_{\Delta.{A},\Delta'}\}:I\ra{}A\{\mathbf{p}_{\Delta,{A}}\circ\mathbf{p}_{\Delta.{A},\Delta'}\} \in \mathcal{L}(\Delta.{A}.\Delta')$$ Note that $\mathrm{der}_{\Delta,A,\Delta'}$ defines a morphism\\ $\Delta.{A}.\Delta'\ra{\mathrm{diag}_{\Delta,{A},\Delta'}}\Delta.{A}.\Delta'.{A}\{\mathbf{p}_{\Delta,{A}}\circ\mathbf{p}_{\Delta.{A},\Delta'}\}:=\langle \mathrm{id}_{\Delta.{A}.\Delta'},\mathrm{der}_{\Delta,A,\Delta'}\rangle $.\\ We will later show that this behaves as a diagonal morphism on ${A}$.

\item $\mathrm{id}_A\in \mathrm{LTerm}(\Delta,A,A)$ is taken to be $\mathrm{id}_A\in\mathcal{L}(\Delta)(A,A)$. Note that this is indeed the neutral element for our semantic linear term substitution operation that we will define shortly.

\item The required morphisms in Int-Weak are interpreted as follows. Suppose we are given $A,\Delta'\in\mathrm{ob}(\mathcal{L}(\Delta))$. We will define a functor
\[
\begin{tikzcd}[ampersand replacement=\&, column sep=huge]
\mathcal{L}(\Delta.\Delta^{\prime}) \arrow[r,"{\mathcal{L}(\langle f,a\rangle)}"] \& \mathcal{L}(\Delta.{A}.\Delta^{\prime}\{\mathbf{p}_{\Delta,{A}}\}),
\end{tikzcd}
\]
where $f$ and $a$ are defined as follows.
\[
\begin{tikzcd}[ampersand replacement=\&, column sep=huge]
\Delta.{A}.\Delta^{\prime}\{\mathbf{p}_{\Delta,{A}}\} \arrow[r,"{f:=\mathbf{p}_{\Delta,{A}}\circ\mathbf{p}_{\Delta.{A},\Delta^{\prime}\{\mathbf{p}_{\Delta,{A}}\}}}"] \& \Delta
\end{tikzcd}
\]
and 
\[
\begin{tikzcd}[ampersand replacement=\&, column sep=large]
I \arrow[r,"{a=\mathbf{v}_{\Delta.{A},\Delta^{\prime}\{\mathbf{p}_{\Delta,{A}}\}}}"] \& \Delta^{\prime}\{f\}=\Delta^{\prime}\{\mathbf{p}_{\Delta,{A}}\circ\mathbf{p}_{\Delta.{A},\Delta^{\prime}\{\mathbf{p}_{\Delta,{A}}\}}\} \;\in \mathcal{L}(\Delta.{A}.\Delta^{\prime}\{\mathbf{p}_{\Delta,{A}}\}).
\end{tikzcd}
\]
\item[8.\&9.] Int-Ty-Subst and Int-Tm-Subst, along a term $\Delta;\cdot \vdash a:A$, are interpreted by the functors $\mathcal{L}(\langle \mathrm{id}_\Delta,a\rangle)=-\{\langle \mathrm{id}_\Delta,a\rangle\}$. Indeed, let $B\in\mathcal{L}(\Delta.{A}.\Delta')$ and $a\in\mathcal{L}(\Delta)(I,A)$. Then, we define the context $\Delta.\Delta'[{a}/x]$  as  $\Delta.(\Delta'\{\langle \mathrm{id}_\Delta,a\rangle\})$ and the type $B[{a}/x]$ as $B\{\langle f,a'\rangle\}$, where
\[
\begin{tikzcd}[ampersand replacement=\&, column sep=huge]
\Delta.\Delta^{\prime}\{\langle \mathrm{id}_\Delta,a\rangle\} \arrow[r,"{\langle f,a^{\prime}\rangle}"] \& \Delta.{A}.\Delta^{\prime}
\end{tikzcd}
\]
is defined from
\[
\begin{tikzcd}[ampersand replacement=\&, column sep=large, row sep=large]
\Delta.\Delta^{\prime}\{\langle \mathrm{id}_\Delta,a\rangle\} \arrow[rr,"{\mathbf{p}_{\Delta,\Delta^{\prime}\{\langle \mathrm{id}_\Delta,a\rangle\}}}"] \arrow[drr,"{f}"'] \&\& \Delta \arrow[d,"{\langle \mathrm{id}_\Delta,a\rangle}"]\\
\&\& \Delta.{A}
\end{tikzcd}
\]
and
\[
\begin{tikzcd}[ampersand replacement=\&, column sep=large]
I \arrow[r,"{a^{\prime}:=\mathbf{v}_{\Delta,\Delta^{\prime}\{\langle \mathrm{id}_\Delta,a\rangle\}}}"] \& \Delta^{\prime}\{f\}=(\Delta^{\prime}\{\langle \mathrm{id}_\Delta,a\rangle\})\{\mathbf{p}_{\Delta,\Delta^{\prime}\{\langle \mathrm{id}_\Delta,a\rangle\}}\}.
\end{tikzcd}
\]
\item[10.] Lin-Tm-Subst is interpreted by composition in $\mathcal{L}(\Delta)$. To be precise, given $b\in \mathcal{L}(\Delta)(\bigotimes\Xi\otimes A,B)$ and $a\in \mathcal{L}(\Delta)(\bigotimes\Xi',A)$, we define $b[a/x]\in \mathcal{L}(\Delta)(\bigotimes\Xi\otimes \bigotimes \Xi',B)$ as $b\circ \mathrm{id}_{\bigotimes\Xi}\otimes a$.

\end{enumerate}
Note that all our substitution rules are interpreted by functors and are therefore associative.

The fact that Int-Var and Int-Weak define compatible (generalised) diagonals and projections is reflected in the fact that $\mathrm{diag}_{}$ and $\mathbf{p}_{}$ obey generalised comonoid laws: 
\[
\begin{tikzcd}[ampersand replacement=\&, column sep=large, row sep=large]
\Delta.{A}.\Delta^{\prime} \arrow[rr,"{\mathrm{diag}_{\Delta,A,\Delta^{\prime}}}"] \arrow[drr,"{\mathrm{id}_{\Delta.{A}.\Delta^{\prime}}}"'] \&\& \Delta.{A}.\Delta^{\prime}.{A}\{\mathbf{p}_{\Delta,{A}.\Delta^{\prime} }\} \arrow[d,"{\mathbf{p}_{\Delta.{A}.\Delta^{\prime},{A\{\mathbf{p}_{\Delta,{A}.\Delta^{\prime} }\}}}}"]\\
\&\&\Delta.{A}.\Delta^{\prime}
\end{tikzcd}
\]
\[
\begin{tikzcd}[ampersand replacement=\&, column sep=tiny, row sep=huge, cells={nodes={font=\scriptsize}}]
\Delta.A.\Delta^{\prime} \arrow[r,"{\mathrm{diag}_{\Delta,A,\Delta^{\prime}}}"] \arrow[d,"{\mathrm{diag}_{\Delta,A,\Delta^{\prime}}}"'] \& \Delta.{A}.\Delta^{\prime}.{A}\{\mathbf{p}_{\Delta,{A}.\Delta^{\prime} }\} \arrow[d,"{\mathrm{diag}_{\Delta,A,\Delta^{\prime}.{A}\{\mathbf{p}_{\Delta,{A}.\Delta^{\prime} }\}}}"]\\
\Delta.{A}.\Delta^{\prime}.{A}\{\mathbf{p}_{\Delta,{A}.\Delta^{\prime} }\} \arrow[r,"\quad{\mathrm{diag}_{\Delta.{A}.\Delta^{\prime},A\{\mathbf{p}_{\Delta,{A}.\Delta^{\prime} }\},\cdot}}"] \& \Delta.{A}.\Delta^{\prime}.{A}\{\mathbf{p}_{\Delta,{A}.\Delta^{\prime} }\}.{A}\{\mathbf{p}_{\Delta,{A}.\Delta^{\prime}.{A}\{\mathbf{p}_{\Delta,{A}.\Delta^{\prime} }\}}\} ,
\end{tikzcd}
\]
where we use $\mathbf{p}_{\Delta,{A}.\Delta'}$ as a shorthand notation for $\mathbf{p}_{\Delta,{A}}\circ\mathbf{p}_{\Delta.{A},\Delta'}$.

Finally, the model supports $I$- and $\otimes$-types. We interpret $I\in\mathrm{LType}(\Delta)$ as the unit object in $\mathcal{L}(\Delta)$ while its term $*$ is interpreted as the identity morphism. Similarly, we interpret $\otimes$ by the monoidal product on the fibres: $*:=\mathrm{id}_I\in\mathcal{L}(\Delta)$, $\mathrm{let}\; t\;\mathrm{be}\; *\;\mathrm{in} \; a:=t\otimes a$, $a\otimes b$ is defined as the tensor product of morphisms in $\mathcal{L}(\Delta)$, and $\mathrm{let}\;t\;\mathrm{be}\; x\otimes y\;\mathrm{in}\;c:= c\circ(\mathrm{id}_{\Xi'}\otimes t)$ (ignoring associators and unitors here). The C- and U-rules are immediate.
\end{proof}
In fact, the converse is also true: we can build a category of this sort from the syntax of ILDTT.
\begin{theorem}[co-Soundness] A tautological model $\widetilde{\mathbb{T}}$ of ILDTT with $I$ and $\otimes$-types defines a strict indexed symmetric monoidal category with comprehension $(\mathcal{C}^\mathbb{T},\mathcal{L}^\mathbb{T},\mathbf{p}^\mathbb{T},\mathbf{v}^\mathbb{T})$.\end{theorem}
\begin{proof}
The main technical difficulty in this proof will be that our category of contexts has context morphisms as morphisms (corresponding to lists of terms of the type theory) while the type theory only talks about individual terms. The same difficulty is also encountered when proving completeness of the categories with families semantics for ordinary DTT. It is sometimes fixed by (conservatively) extending the type theory to also talk about context morphisms explicitly. See e.g. \cite{pitts2001categorical}.

\begin{enumerate}
\item[1.] We define $\mathrm{ob}(\mathcal{C}^\mathbb{T}):=\mathrm{ICtxt}$, modulo $\alpha$-equivalence, and write $\Delta.{A}:=\Delta,x:A$. The designated object will be $\cdot$ (from C-Emp), which will automatically become a terminal object because of our definition of a morphism of $\mathcal{C}^\mathbb{T}$ (context morphism). Indeed, we define morphisms in $\mathcal{C}^\mathbb{T}$ by induction as follows.

Start by defining $\mathcal{C}^\mathbb{T}(\Delta',\cdot):=\{\langle\rangle\}$ and for $\Delta\in\mathrm{ICtxt}$ that is not of the form $\Delta''.{A}$, define $\mathcal{C}^\mathbb{T}(\Delta',\Delta)=\{\mathrm{id}_\Delta\}$ if $\Delta'=\Delta$ and $\mathcal{C}^\mathbb{T}(\Delta',\Delta)=\emptyset$ otherwise.

Then, by induction on the length $n$ of $\Delta=x_1:A_1,\ldots,x_n:A_n$, we define 
$$\mathcal{C}^\mathbb{T}(\Delta',\Delta.{A_{n+1}}):=\Sigma_{f\in\mathcal{C}^\mathbb{T}(\Delta',\Delta)}\mathrm{LTerm}(\Delta',\cdot,A_{n+1}[f/x]),$$
where $A_{n+1}[f/x]$ is defined, using Int-Ty-Subst, to be the (syntactic operation of) parallel substitution (see \cite{hofmann1997syntax}, section 2.4) of the list $f_1,\ldots, f_n$ of linear terms $\Delta';\cdot \vdash f_i:A_i[f_1/x_1,\ldots,f_{i-1}/x_{i-1}]$ that $f$ consists of, for the variables $x_1,\ldots, x_n$ in $\Delta$.

Note that, in particular, according to Int-Var, $\mathrm{LTerm}({A_1}.\ldots .{A_n},\cdot,{A_i})$ contains a term\\ $\mathrm{der}_{{A_1}.\ldots.{A_{i-1}},{A_i},{A_{i+1}}.\ldots.{A_n}}$, which allows us to define, inductively,
$$\mathbf{p}_{{A_1}.\ldots.{A_n},\cdot}:=\langle\rangle\in\mathcal{C}^\mathbb{T}({A_1}.\ldots.{A_n},\cdot)$$
$$\mathbf{p}_{{A_1}.\ldots.{A_n},{A_1}.\ldots.{A_i}}:=\mathbf{p}_{{A_1}.\ldots.{A_n},{A_1}.\ldots.{A_{i-1}}},\mathrm{der}_{{A_1}.\ldots.{A_{i-1}},{A_i},{A_{i+1}}.\ldots.{A_n}}\in\mathcal{C}^\mathbb{T}({A_1}.\ldots.{A_n},{A_1}.\ldots.{A_i})$$
In particular, we define identities in $\mathcal{C}^\mathbb{T}$ from these: $\mathrm{id}_{{A_1}.\ldots.{A_n}}:=\mathbf{p}_{{A_1}.\ldots.{A_n},{A_1}.\ldots.{A_n}}$. We will also use these 'projections' in item 3 to define the comprehension schema.

We define composition in $\mathcal{C}^\mathbb{T}$ by induction. Let ${B_1}.\ldots.{B_m}=\Delta'\ra{f=f_1,\ldots,f_n}\Delta={A_1}.\ldots.{A_n}$ and $\Delta''\ra{g=g_1,\ldots,g_m}\Delta'$. Then we define $(f_1,\ldots,f_{n-1},f_n)\circ g:=(f_1,\ldots, f_{n-1})\circ g,f_n[g/x]$, where $f_n[g/x]$ denotes the parallel substitution of $g=g_1,\dots,g_m$ for the free variables $x_1,\ldots,x_m$ in $f_n$, using Int-Tm-Subst. Note that associativity of composition comes from the associativity of substitution that is implicit in the syntax, while the identity morphism we defined acts as a neutral element for our composition.

\item[2.] Define $\mathrm{ob}(\mathcal{L}^\mathbb{T}(\Delta)):=\mathrm{LCtxt}(\Delta)$ and $\mathcal{L}^\mathbb{T}(\Delta)(\Xi,\Xi'):=\mathrm{LTerm}(\Delta,\Xi,\bigotimes \Xi')$. Composition is defined through Lin-Tm-Subst and $\otimes$-E. Identities are given by Lin-Var. The monoidal unit is given by $\cdot\in\mathrm{LCtxt}(\Delta)$, while the monoidal product $\otimes$ on objects is given by context concatenation. The monoidal product $\otimes$ on morphisms is given by $\otimes$-I. Note that the associators and unitors follow from the associative and unital laws for the commutative monoid of contexts together with $\otimes$-C and $\otimes$-U. (Note that the rules for $\otimes$ give us an isomorphism between an arbitrary context $\Xi$ and the one-type-context $\bigotimes \Xi$, while the rules for $I$ do the same for $\cdot$ and $I$.)

We define $\mathcal{L}^\mathbb{T}(f)$ on objects by parallel substitution in each type in a linear context, via Int-Ty-Subst, and on morphisms by parallel substitution, via Int-Tm-Subst. Note that functoriality is given by implicit properties of the syntax, such as associativity of substitution. Note that this defines a strong symmetric monoidal functor. We conclude that $\mathcal{L}^\mathbb{T}$ is a functor $\mathcal{C}^\mathbb{T}{}^{op}\ra{}\mathrm{SMCat}$.

\item[3.] We define the following comprehension schema on $\mathcal{L}^\mathbb{T}$. Suppose $\Delta\in\mathcal{C}^\mathbb{T}$ and $A\in\mathcal{L}^\mathbb{T}(\Delta)$.

Define $\Delta.{A}\ra{\mathbf{p^\mathbb{T}}_{\Delta,A}}\Delta$ as $\mathbf{p}_{\Delta,A}$ from 1. and $I\ra{\mathbf{v^\mathbb{T}}_{\Delta,{A}}} A\{\mathbf{p}_{\Delta,{A}}\}$ (through Int-Var) as\\ $\mathrm{der}_{A}\in\mathrm{LTerm}(\Delta.{A},\cdot,A)=\mathcal{L}^\mathbb{T}(\Delta.{A})(I,A\{\mathbf{p^\mathbb{T}}_{\Delta,{A}}\})$.

Suppose we are given $\Delta'\ra{f}\Delta$ and $a\in\mathcal{L}^\mathbb{T}(\Delta')(I,A\{f\})=\mathrm{LTerm}(\Delta',\cdot,A[f/x])$. Then, by definition of the morphisms in $\mathcal{C}^\mathbb{T}$, there is a unique morphism\\ $\langle f,a\rangle:= f,a\in\mathcal{C}^\mathbb{T}(\Delta',\Delta.{A}):=\Sigma_{f\in\mathcal{C}^\mathbb{T}(\Delta',\Delta)}\mathrm{LTerm}(\Delta',\cdot,A[f/x])$ such that $\mathbf{p}^\mathbb{T}_{\Delta,{A}}\circ\langle f,a\rangle =f$ and $\mathbf{v}^\mathbb{T}_{\Delta,{A}}\{\langle f,a\rangle \}=a$. The uniqueness follows from the fact that $\mathbf{p}^\mathbb{T}_{\Delta,{A}}\circ - $ and $\mathbf{v^\mathbb{T}}_{\Delta,{A}}\{-\}$ are the two (dependent) projections of the $\Sigma$-type (in $\mathrm{Set}$) that defines this homset.
\end{enumerate}
\end{proof}

\begin{theorem}[Completeness] The construction described in 'co-Soundness' followed by the one described in 'Soundness' is the identity (up to categorical equivalence\footnote{The correct formal statement here would be that co-soundness followed by soundness (both of which define 2-functors between the 2-category of tautological models of ILDTT and the 2-category of strict indexed symmetric monoidal categories with comprehension) is 2-equivalent to the identity.}); that is, strict indexed symmetric monoidal categories with comprehension provide a complete semantics for ILDTT with $I$- and $\otimes$-types\footnote{It is easy to see that, similarly, indexed symmetric multicategories with comprehension form a complete semantics for ILDTT, possibly without $I$- and $\otimes$-types.}.
\end{theorem}
\begin{proof}This is a standard exercise.\end{proof}
\begin{theorem}[Failure of co-Completeness] The construction described in 'Soundness' followed by the one described in 'co-Soundness' is not equivalent to the identity; that is, co-Completeness fails (as for the categories with families semantics for DTT).
\end{theorem}
\begin{proof}
Indeed, if we start with a strict indexed symmetric monoidal category with comprehension, construct the corresponding tautological model $\widetilde{\mathbb{T}}$ and then construct its syntactic category, we have effectively thrown away all the non-trivial morphisms into objects that are not of the form $\Delta.{A}$.

Of course, we can easily obtain a co-complete model theory by putting this extra restriction on our models. Alternatively, perhaps more naturally from a categorical point of view, we can take the obvious (see e.g. \cite{pitts2001categorical}) conservative extension of our syntax by also talking about context morphisms (corresponding to morphisms in our base category). In that case, we would obtain an actual internal language for strict indexed symmetric monoidal categories with comprehension. This also has the advantage that we can easily obtain an internal language for strict indexed monoidal categories by dropping the axioms Int-C-Ext, Int-C-Ext-Eq, and Int-Var, which correspond to the comprehension schema. We have not chosen this route because it would mean that the syntax does not fit as well with what has been considered so far in the syntactic tradition.\end{proof}
\begin{corollary}[Relation with DTT and ILTT] A model $(\mathcal{C},\mathcal{L},\mathbf{p},\mathbf{v})$ of ILDTT with $I$- and $\otimes$-types defines a model $\mathcal{I}$ of DTT that should be thought of as the intuitionistic content of the linear type theory. This will become even clearer through our treatment of $!$-types and in the examples we treat.

Moreover, it defines a model of ILTT with $I$- and $\otimes$-types (i.e. a symmetric monoidal category) in every context.

Conversely, it is easily seen that every syntactic model\footnote{i.e. a model where we do not have any non-trivial morphisms into contexts that are not built from $\cdot$ by appending types.} of DTT can be obtained this way, up to equivalence, from a syntactic model of ILDTT (we take the same constants and axioms: effectively the same theory but in a system without contraction and weakening) and that every model of ILTT can be embedded in a model of ILDTT. (As we will see later, we can cofreely add type dependency on $\mathrm{Set}$.)\end{corollary}

\subsubsection*{Semantic Type Formers}
\begin{theorem}[Semantic type formers]\label{thm:semtype} For the other type formers, a model of ILDTT with $I$- and $\otimes$-types (a strict indexed symmetric monoidal category with comprehension) has the following properties.
\begin{enumerate}
\item It supports $\Sigma$-types iff all the pullback functors $\mathcal{L}(\mathbf{p}_{\Delta,{A}})$ have left adjoints $\Sigma_{!{A}}$ that satisfy the Beck-Chevalley condition\footnote{Remember that the Beck-Chevalley condition for a pullback square \begin{center}
\begin{tikzcd}[ampersand replacement=\&, column sep=large, row sep=large]
A \arrow[r,"{h}"] \arrow[d,"{f}"'] \& B \arrow[d,"{k}"]\\
C \arrow[r,"{g}"'] \& D
\end{tikzcd}
\end{center}
corresponds to the statement that the obvious morphism (from commutativity of the pullback square, unit, and counit) $f_!h^*\ra{}f_!h^*k^*k_!\ra{\cong}f_!f^*g^*k_!\ra{}g^*k_!$ is an isomorphism, where we write $f^*$ for $\mathcal{L}(f)$ and $f_!$ for its left adjoint. In the case where $f^*$ has a right adjoint (as for the $\Pi$-type), $f_*$, we mean by the dual Beck-Chevalley condition that the obvious morphism $g^*k_*\ra{}g^*k_*h_*h^*\ra{\cong}g^*g_*f_*h^*\ra{}f_*h^*$ is an isomorphism.}\footnote{This condition should be thought of as the analogue, for $\Sigma$- and $\Pi$-types, of the condition on the substitution functors preserving the appropriate categorical structure for other type formers. It says that, in a sense, $\Sigma$- and $\Pi$-types are preserved under substitution.} for pullback squares in $\mathcal{C}$ of the following form,\footnote{Recall that $\mathbf{q}_{f,{B}}:=\langle f\circ\mathbf{p}_{\Delta',{B\{f\}}},\mathbf{v}_{\Delta',{B\{f\}}} \rangle$ and that this square is indeed a pullback.}.
\[
\begin{tikzcd}[ampersand replacement=\&, column sep=large, row sep=large]
\Delta^{\prime}.{B\{f\}} \arrow[r,"{\mathbf{q}_{f,{B}}}"] \arrow[d,"{\mathbf{p}_{\Delta^{\prime},{B\{f\}}}}"'] \arrow[dr, phantom,"{(*)}" description] \& \Delta.{B} \arrow[d,"{\mathbf{p}_{\Delta,{B}}}"]\\
\Delta^{\prime} \arrow[r,"{f}"'] \& \Delta,
\end{tikzcd}
\]
and that satisfy Frobenius reciprocity\footnote{Frobenius reciprocity expresses compatibility of $\Sigma$ and $\otimes$, which is reasonable if we want a reading of $\Sigma$ as a generalisation of $\otimes$. If one wants to drop Frobenius reciprocity in the semantics, it is easy to see that the syntactic counterpart is setting $\Xi'\equiv\cdot$ in the $\Sigma$-E-rule.} in the sense that the canonical morphism $$\Sigma_{!{A}}(\Xi'\{\mathbf{p}_{\Delta,{A}}\}\otimes B)\ra{} \Xi'\otimes \Sigma_{!{A}}B$$ is an isomorphism, for all $\Xi'\in\mathcal{L}(\Delta)$, $B\in\mathcal{L}(\Delta.{A})$.
\item It supports $\Pi$-types iff all the pullback functors $\mathcal{L}(\mathbf{p}_{\Delta,{A}})$ have right adjoints $\Pi_{!{A}}$ that satisfy the dual Beck-Chevalley condition for pullbacks of the form $(*)$.
\item It supports $\multimap$-types iff $\mathcal{L}$ factors over the category of symmetric monoidal closed categories and closed strong monoidal functors.
\item It supports $\top$-types and $\&$-types iff $\mathcal{L}$ factors over the category of Cartesian categories with a symmetric monoidal structure and their homomorphisms.
\item It supports $0$-types and $\oplus$-types iff $\mathcal{L}$ factors over the category $\mathrm{dSMcCCat}$ of co-Cartesian categories with a distributive\footnote{Note that in light of theorem \ref{thm:pisigmainf}, the demand of distributivity here is essentially the same phenomenon as the demand of Frobenius reciprocity for $\Sigma$-types.} symmetric monoidal structure and their homomorphisms.
\item If it supports $\multimap$-types\footnote{Actually, we only need this for the 'if'. The 'only if' always holds. To make the 'if' work as well, in the absence of $\multimap$-types, we have to restrict $!$-E to the case where $\Xi'\equiv\cdot$.}, then it supports $!$-types iff all the comprehension functors $\mathcal{L}(\Delta)\ra{M_\Delta}\mathcal{I}(\Delta)$ have a left adjoint $\mathcal{I}(\Delta)\ra{L_\Delta}\mathcal{L}(\Delta)$ in the 2-category $\mathrm{SMCat}$ of symmetric monoidal categories, lax symmetric monoidal functors, and monoidal natural transformations\footnote{That is, a symmetric lax monoidal left adjoint functor $L_\Delta$ such that an inverse for its lax structure is given by the oplax structure on $L_\Delta$ coming from the lax structure on $M_\Delta$. Put differently, $L_\Delta$ is a left adjoint functor to $M_\Delta$ and is a strong monoidal functor in a way that is compatible with the lax structure on $M_\Delta$.} and, compatibly with substitution, for all $\Delta'\ra{f}\Delta\in \mathcal{C}$ the following square commutes, making $L_-$ a morphism of indexed categories:
\[
\begin{tikzcd}[ampersand replacement=\&, column sep=huge, row sep=large]
\mathcal{L}(\Delta) \arrow[r,"{\mathcal{L}(f)}"] \& \mathcal{L}(\Delta^{\prime})\\
\mathcal{I}(\Delta) \arrow[u,"{L_\Delta}"] \arrow[r,"{\begin{gathered}\mathcal{I}(f)=f^*\\[-0.5ex]\textnormal{(=pullback along }f\textnormal{)}\end{gathered}}"'] \& \mathcal{I}(\Delta^{\prime}) \arrow[u,"{L_{\Delta^{\prime}}}"'].
\end{tikzcd}
\]
\quad\\
Then the linear exponential comonad $!_\Delta:=L_\Delta\circ M_\Delta:\mathcal{L}(\Delta)\ra{}\mathcal{L}(\Delta)$ is our interpretation of the comodality $!$ in the context $\Delta$.
\item If it supports $\multimap$-types, then it supports $\mathrm{Id}$-types iff for all $A\in\mathrm{ob}\;\mathcal{L}(\Delta)$, we have left adjoints $\mathrm{Id}_{!A}\dashv -\{\mathrm{diag}_{\Delta,{A}}\}$ that satisfy a Beck-Chevalley condition: $\mathrm{Id}_{!A\{f\}}\circ \mathcal{L}(\mathbf{q}_{f,{A}})\ra{}\mathcal{L}(\mathbf{q}_{\mathbf{q}_{f,{A}},{A\{\mathbf{p}_{\Delta,A}\}}})\circ \mathrm{Id}_{!A}$ is an isomorphism.\\
Here, \footnotesize\mbox{\begin{tikzcd}[ampersand replacement=\&, column sep=huge, baseline=-0.6ex]
\Delta.A \arrow[r,"{\mathrm{diag}_{\Delta,A}:=\langle \mathrm{id}_{\Delta.A},\mathbf{v}_{\Delta,A}\rangle}"] \& \Delta.A.A\{\mathbf{p}_{\Delta,A}\}
\end{tikzcd}}\normalsize.
\end{enumerate}\end{theorem}
\begin{proof}
\begin{enumerate}
\item Assume our model supports $\Sigma$-types. We will show the claimed adjunction. The morphism from left to right is provided by $\Sigma$-I. The morphism from right to left is provided by $\Sigma$-E. $\Sigma$-C and $\Sigma$-U say exactly that these are mutually inverse. Naturality corresponds to the compatibility of $\Sigma$-I and $\Sigma$-E with substitution.
\[
\begin{tikzcd}[ampersand replacement=\&, column sep=large, row sep=small]
c' \arrow[r, maps to] \& c'\{\mathbf{p}_{\Delta,{A}}\}\circ\langle \mathrm{diag}_{\Delta,{A},\cdot},\mathrm{id}_B\rangle \\
\mathcal{L}(\Delta)(\Sigma_{!{A}}B,C) \arrow[r, shift left=1.5ex] \arrow[r, phantom,"{\cong}" description] \& \mathcal{L}(\Delta.{A})(B,C\{\mathbf{p}_{\Delta,{A}}\}) \arrow[l, shift left=1.5ex]\\
\mathrm{let}\;z\;\mathrm{be}\;  !{x} \otimes y\;\mathrm{in}\;c \& c \arrow[l, maps to]
\end{tikzcd}
\]

We show how the morphism from left to right arises from $\Sigma$-I.\\
\\
\small
\AxiomC{}
\RightLabel{Int-Var}
\UnaryInfC{$\Delta,x:A;\cdot \vdash x:A$}
\AxiomC{}
\RightLabel{Lin-Var}
\UnaryInfC{$\Delta,x:A;w:B\vdash w:B$}
\RightLabel{$\Sigma$-I}
\BinaryInfC{$\Delta,x:A;w:B\vdash  !{x} \otimes w:\Sigma_{!x:{!A}}B$}

\AxiomC{$\Delta;z:\Sigma_{!x:{!A}}B\vdash c': C$}
\RightLabel{Int-Weak}
\UnaryInfC{$\Delta,x:A;z:\Sigma_{!x:{!A}}B\vdash c': C$}
\RightLabel{Lin-Tm-Subst}
\BinaryInfC{$\Delta,x:A;w:B\vdash c'[ !{x}\otimes w/z]:C$}
\DisplayProof

\normalsize
\quad\\
\\
We show how the morphism from right to left is exactly $\Sigma$-E (with $\Xi'\equiv\cdot$, $\Xi\equiv z:\Sigma_{!x:{!A}}B$, $t\equiv z$).\\
\\
\AxiomC{$\Delta;\cdot \vdash C\;\mathrm{type}$}
\AxiomC{}
\RightLabel{Lin-Var}
\UnaryInfC{$\Delta;z:\Sigma_{!x:{!A}}B \vdash z:\Sigma_{!x:{!A}}B$}
\AxiomC{$\Delta,x:A;y:B\vdash c:C$}
\RightLabel{$\Sigma$-E}
\TrinaryInfC{$\Delta;z:\Sigma_{!x:{!A}}B\vdash \mathrm{let}\;z\;\mathrm{be}\;  !{x} \otimes y \;\mathrm{in}\;c:C$}
\DisplayProof
\\
\\
We show how Frobenius reciprocity can be proved in our type system (relying in particular on the form of the $\Sigma$-E-rule\footnote{To be precise, we will see that Frobenius reciprocity is validated because we allow dependency on $\Xi'$ in the $\Sigma$-E-rule. Conversely, it is easy to see that we can prove Frobenius reciprocity in our model if we have (semantic) $\multimap$-types, as this allows us to remove the dependency on $\Xi'$ in $\Sigma$-E.}).
\begin{lemma}[Frobenius reciprocity] The canonical morphism $$\Sigma_{!{A}}(\Xi'\{\mathbf{p}_{\Delta,{A}}\}\otimes B)\ra{f} \Xi'\otimes \Sigma_{!{A}}B$$ is an isomorphism, for all $\Xi'\in\mathcal{L}(\Delta)$, $B\in\mathcal{L}(\Delta.{A})$.
\end{lemma}
\begin{proof} We first show how to construct the morphism $f$ intended here.\\
\\
\tiny
\AxiomC{}
\RightLabel{Lin-Var}
\UnaryInfC{$\Delta;x':\Sigma_{!x:{!A}}(\Xi'\otimes B)\vdash x':\Sigma_{!x:{!A}}(\Xi'\otimes B)$}

\AxiomC{}
\RightLabel{Lin-Var}
\UnaryInfC{$\Delta,x:A;z:\Xi'\vdash z:\Xi'$}

\AxiomC{}
\RightLabel{Int-Var}
\UnaryInfC{$\Delta,x:A;\cdot \vdash x:A$}
\AxiomC{}
\RightLabel{Lin-Var}
\UnaryInfC{$\Delta;y:B\vdash y:B$}
\RightLabel{$\Sigma$-I}
\BinaryInfC{$\Delta,x:A;y:B\vdash  !{x} \otimes y:\Sigma_{!x:{!A}}B$}
\RightLabel{$\otimes$-I}
\BinaryInfC{$\Delta,x:A;z:\Xi', y:B\vdash z\otimes  ! {x} \otimes y:  \Xi'\otimes \Sigma_{!x:{!A}}B$}
\RightLabel{$\otimes$-E}
\UnaryInfC{$\Delta,x:A;w: \Xi'\otimes B\vdash \mathrm{let}\; w\;\mathrm{be}\;z\otimes y\;\mathrm{in}\;z\otimes  ! {x} \otimes y :\Xi'\otimes \Sigma_{!x:{!A}}B$}
\RightLabel{$\Sigma$-E}
\BinaryInfC{$\Delta;x':\Sigma_{!x:{!A}}(\Xi'\otimes B)\vdash f: \Xi'\otimes \Sigma_{!x:{!A}}B$}
\DisplayProof
\normalsize\\
\\
We now construct its inverse. Call it $g$.\\
\\
\tiny
\AxiomC{}
\RightLabel{Lin-Var}
\UnaryInfC{$\Delta;y_2:\Sigma_{!x:{!A}}B\vdash y_2:\Sigma_{!x:!{A}}B$}
\AxiomC{}
\RightLabel{Int-Var}
\UnaryInfC{$\Delta,x:A;\cdot \vdash x:A$}
\AxiomC{}
\RightLabel{Lin-Var}
\UnaryInfC{$\Delta;y_1:\Xi'\vdash y_1:\Xi'$}
\AxiomC{}
\RightLabel{Lin-Var}
\UnaryInfC{$\Delta;y:B\vdash y:B$}
\RightLabel{$\otimes$-I}
\BinaryInfC{$\Delta,x:A;y_1:\Xi',y:B\vdash y_1\otimes y:\Xi'\otimes B$}
\RightLabel{$\Sigma$-I}
\BinaryInfC{$\Delta,x:A;y_1:\Xi',y:B\vdash  !{x} \otimes  y_1\otimes y:\Sigma_{!x:{!A}}(\Xi'\otimes B)$}
\RightLabel{$\Sigma$-E}
\BinaryInfC{$\Delta;y_1:\Xi',y_2:\Sigma_{!x:!{A}}B\vdash \mathrm{let}\; y_2\;\mathrm{be}\; !{x} \otimes y\; \mathrm{in}\;  !{x} \otimes  y_1\otimes y:\Sigma_{!x:!A}(\Xi'\otimes B)$}
\RightLabel{$\otimes$-E}
\UnaryInfC{$\Delta;y':\Xi'\otimes \Sigma_{!x:{!A}}B\vdash g:\Sigma_{!x:{!A}}(\Xi'\otimes B)$}
\DisplayProof\footnote{The use of $\Sigma$-E is precisely where Frobenius reciprocity comes in, because of the factor $\Xi'$ in the $\Sigma$-E-rule.}
\quad
\normalsize
\\
\\
We leave it to the reader to verify that these morphisms are mutually inverse in the sense that $$\Delta;x':\Sigma_{!x:{!A}}(\Xi'\otimes B)\vdash g[f/y']\equiv x':\Sigma_{!x:{!A}}(\Xi'\otimes B)\txt{and} \Delta;y':\Xi'\otimes \Sigma_{!x:{!A}}B \vdash f[g/x']\equiv y':\Xi'\otimes \Sigma_{!x:!{A}}B .$$\end{proof}
\quad\\
\\
For the converse, we show how to obtain $\Sigma$-I from our morphism from left to right:\\
\\
\footnotesize
\AxiomC{}
\RightLabel{Lin-Var}
\UnaryInfC{$\Delta;z:\Sigma_{!x:{!A}}B\vdash z:\Sigma_{!x:{!A}}B$}
\RightLabel{"left to right"}
\UnaryInfC{$\Delta,x:A;w:B\vdash  !{x}\otimes w:\Sigma_{!x:{!A}}B$}
\AxiomC{$\Delta;\cdot \vdash a:A$}
\RightLabel{Int-Tm-Subst}
\BinaryInfC{$\Delta;w:B\vdash  !{a}\otimes w:\Sigma_{!x:{!A}}B$}
\AxiomC{$\Delta;\Xi\vdash b:B[{a}/x]$}
\RightLabel{Lin-Tm-Subst}
\BinaryInfC{$\Delta;\Xi\vdash !{a}\otimes b:\Sigma_{!x:{!A}}B$}
\DisplayProof
\normalsize
\\
\\
We show how to obtain $\Sigma$-E from our morphism from right to left, using Frobenius reciprocity. In particular, note that we do need Frobenius reciprocity.\\
\\
\footnotesize
\AxiomC{$\Delta;\cdot\vdash C\;\mathrm{type}$}
\AxiomC{$\Delta,x:A;\Xi', y:B\vdash c:C$}
\RightLabel{$\otimes$-E}
\UnaryInfC{$\Delta,x:A;y:\Xi'\otimes B\vdash c:C$}
\RightLabel{"right to left"}
\BinaryInfC{$\Delta;z:\Sigma_{!x:{!A}}(\Xi'\otimes B)\vdash \mathrm{let}\;z\;\mathrm{be}\;  !{x} \otimes y \;\mathrm{in}\;c:C$}
\RightLabel{Frobenius reciprocity}
\UnaryInfC{$\Delta;z:(\Xi'\otimes \Sigma_{!x:{!A}}B)\vdash \mathrm{let}\;z\;\mathrm{be}\;  !{x} \otimes y \;\mathrm{in}\;c:C$}
\RightLabel{Lin-Tm-Subst,$\otimes$-I,2$\times$Lin-Var}
\UnaryInfC{$\Delta;z_1:\Xi', z_2:\Sigma_{!x:!{A}}B\vdash \mathrm{let}\;z_1\otimes z_2\;\mathrm{be}\;  !{x} \otimes y \;\mathrm{in}\;c:C$}
\AxiomC{$\Delta;\Xi\vdash t: \Sigma_{!x:{!A}}B$}
\RightLabel{Lin-Tm-Subst}
\BinaryInfC{$\Delta;z_1:\Xi',\Xi\vdash (\mathrm{let}\;z_1\otimes z_2\;\mathrm{be}\;  !{x} \otimes y \;\mathrm{in}\;c)[t/z_2]:C$}
\DisplayProof
\normalsize
\\
\\
As usual, the Beck-Chevalley condition says precisely that $\Sigma$-types commute with substitution, as dictated by the type theory.

\item Assume our model supports $\Pi$-types. We will show the claimed adjunction. The morphism from left to right is provided by $\Pi$-I (indeed, it is exactly the introduction rule), and the one from right to left by $\Pi$-E. $\Pi$-C and $\Pi$-U say exactly that these are mutually inverse. Naturality corresponds to the compatibility of $\Pi$-I and $\Pi$-E with substitution.
\[
\begin{tikzcd}[ampersand replacement=\&, column sep=large, row sep=small]
b \arrow[r, maps to] \& \lambda_{!x:{!A}} b \\
\mathcal{L}(\Delta.{A})(\Xi\{\mathbf{p}_{\Delta,{A}}\},B) \arrow[r, shift left=1.5ex] \arrow[r, phantom,"{\cong}" description] \& \mathcal{L}(\Delta)(\Xi,\Pi_{!x:{!A}} B) \arrow[l, shift left=1.5ex]\\
f(!x) \& f \arrow[l, maps to].
\end{tikzcd}
\]
We show how we obtain the definition of $f(!x)$ from $\Pi$-E.\\
\\
\AxiomC{}
\RightLabel{Int-Var}
\UnaryInfC{$\Delta,x:A;\cdot \vdash x:A$}
\AxiomC{$\Delta;\Xi \vdash f:\Pi_{!x:!A}B$}
\RightLabel{Int-Weak}
\UnaryInfC{$\Delta,x:A;\Xi \vdash f:(\Pi_{!x:!A}B)$}
\RightLabel{$\Pi$-E}
\BinaryInfC{$\Delta,x:A;\Xi\vdash f(!x):B$}
\DisplayProof

For the converse, we have to show that we can recover $\Pi$-E from the definition of $f(!x)$.

\AxiomC{$\Delta;\cdot\vdash a:A$}
\AxiomC{$\Delta;\Xi \vdash f:\Pi_{!x:!A}B$}
\RightLabel{Definition $f(!x)$}
\UnaryInfC{$\Delta,x:A;\Xi\vdash f(!x):B$}
\RightLabel{Int-Tm-Subst}
\BinaryInfC{$\Delta;\Xi\vdash f(!x)[{a}/x]:B[{a}/x]$}
\UnaryInfC{$\Delta;\Xi\vdash f(!{a}):B[{a}/x]$}
\DisplayProof

This shows that individual $\Pi$-types correspond to right adjoint functors to substitution along projections. The type theory dictates that $\Pi$-types interact well with substitution. This corresponds to the dual Beck-Chevalley condition, as usual.

\item From the categorical semantics of (non-dependent) linear type theory (see e.g. \cite{bierman1994intuitionistic} for a comprehensive account) we know that $\multimap$-types correspond to monoidal closure of the category of contexts. The extra feature in dependent linear type theory is that the syntax dictates that the type formers are compatible with substitution. This means that we also have to restrict the functors $\mathcal{L}(f)$ to preserve the relevant categorical structure.

\item The same argument applies.
\item The same argument applies.
\item Assume that we have $!$-types. We will define a left adjoint $L_\Delta\dashv M_\Delta$ as $L_\Delta\mathbf{p}_{\Delta,{A}}:=!A$ (this is easily seen to be well-defined up to isomorphism, so we can use AC for a definition on the nose) and, noting that every morphism $\mathbf{p}_{\Delta,{A}}\ra{}\mathbf{p}_{\Delta,{B}}$ in $\mathcal{C}/\Delta$ is of the form $\langle \mathbf{p}_{\Delta,{A}},b\rangle$ for some unique $I\ra{b}B\{\mathbf{p}_{\Delta,{A}}\}\in\mathcal{L}(\Delta.{A})$, we define $L_\Delta$ as acting on $b$ as the map obtained from

\AxiomC{$\Delta,x:A;\cdot \vdash b:B$}
\RightLabel{!-I}
\UnaryInfC{$\Delta,x:A;\cdot \vdash !b:!B$}
\AxiomC{}
\RightLabel{Lin-Var}
\UnaryInfC{$\Delta;y:!A\vdash y:!A$}
\RightLabel{!-E}
\BinaryInfC{$\Delta;y:!A \vdash\mathrm{let}\; y\;\mathrm{be}\;!x\;\mathrm{in}\; !b:!B$}
\DisplayProof
\\
\\
which indeed gives us $L_\Delta(\langle \mathbf{p}_{\Delta,{A}},b\rangle)\in\mathcal{L}(\Delta)(!A,!B)$. Note that $L_\Delta$ is strong monoidal, as the rules for $!$ define a natural bijection between terms $\Delta,x:A,y:B;\cdot \vdash t:C$ and $\Delta;x':!A,y':!B\vdash t':C$. In terms of the model, this gives a natural bijection $\mathcal{L}(\Delta)(L_\Delta(M_\Delta A \times M_\Delta B),C)\cong \mathcal{L}(\Delta)(!D,C)\cong \mathcal{L}(\Delta.D)(I,C)\cong \mathcal{L}(\Delta.A.B)(I,C)\cong \mathcal{L}(\Delta)(!A\otimes !B,C)$, where we write $D$ for an object such that $M_\Delta D=M_\Delta A \times M_\Delta B$ (which exists if the product exists), so strong monoidality follows by the Yoneda lemma. (The reader can verify that the oplax structure on $L_\Delta$ corresponds to the lax structure on $M_\Delta$.)\\
\\
We exhibit the adjunction by the following isomorphism of hom-sets, where the morphism from left to right comes from $!$-I and the one from right to left comes from $!$-E.\footnotesize
\[
\begin{tikzcd}[ampersand replacement=\&, column sep=small, row sep=small, cells={nodes={font=\scriptsize}}]
b \arrow[r, maps to] \& b[!x/x^{\prime}]  \\
\mathcal{L}(\Delta)(L_\Delta\mathbf{p}_{\Delta,{A}},B)=\mathcal{L}(\Delta)(!A,B) \arrow[r, shift left=1.5ex] \arrow[r, phantom,"{\cong}" description] \& \mathcal{L}(\Delta.A)(I,B\{\mathbf{p}_{\Delta,A}\})\cong \mathcal{C}/\Delta(\mathbf{p}_{\Delta,{A}},\mathbf{p}_{\Delta,{B}})=\mathcal{I}(\Delta)(\mathbf{p}_{\Delta,{A}},M_\Delta B) \arrow[l, shift left=1.5ex]\\
\mathrm{let}\; y\;\mathrm{be}\; !x\;\mathrm{in}\; b^{\prime} \& b^{\prime} \arrow[l, maps to]
\end{tikzcd}
\]\normalsize

We show how to construct the morphism from left to right, using $!$-I.\\
\\
\AxiomC{$\Delta;x':!A\vdash b:B$}
\RightLabel{Int-Weak}
\UnaryInfC{$\Delta,x:A;x':!A\vdash b:B$}
\AxiomC{}
\RightLabel{Int-Var}
\UnaryInfC{$\Delta,x:A;\cdot\vdash x:A$}
\RightLabel{$!$-I}
\UnaryInfC{$\Delta,x:A;\cdot \vdash !x:!A$}
\RightLabel{Lin-Tm-Subst}
\BinaryInfC{$\Delta,x:A;\cdot\vdash b[!x/x']:B$}
\DisplayProof

We show how to construct the morphism from right to left, using $!$-E. Suppose we are given $b'\in\mathcal{L}(\Delta.{A})(I,B\{\mathbf{p}_{\Delta,A}\})$. From this, we produce a morphism in $\mathcal{L}(\Delta)(!A,B)$ as follows.\\
\\
\AxiomC{}
\RightLabel{Lin-Var}
\UnaryInfC{$\Delta;y:!A\vdash y:!A$}
\AxiomC{$\Delta,x:A;\cdot \vdash b':B$}
\RightLabel{!E}
\BinaryInfC{$\Delta;y:!A\vdash \mathrm{let}\; y\;\mathrm{be}\; !x\;\mathrm{in}\; b':B$}
\DisplayProof
\quad\\
\\
\\
We leave it to the reader to verify that these morphisms are mutually inverse, according to $!$-C and $!$-U.\\
\\
Conversely, suppose we have a strong monoidal left adjoint $L_\Delta\dashv M_\Delta$. We define, for $A\in\mathrm{ob}(\mathcal{L}(\Delta))$, $!A:=L_\Delta M_\Delta(A)$.

We verify that $!$-I can be derived from the homset morphism from left to right:\\
\\
\AxiomC{}
\RightLabel{Lin-Var}
\UnaryInfC{$\Delta;x':!A\vdash x':!A$}
\RightLabel{"left to right"}
\UnaryInfC{$\Delta,x:A;\cdot \vdash !x:!A$}
\AxiomC{$\Delta;\cdot \vdash a:A$}
\RightLabel{Int-Tm-Subst}
\BinaryInfC{$\Delta;\cdot \vdash !x[{a}/x]:!A$}
\DisplayProof
\quad\\
\\
We verify that, in the presence of $\multimap$-types, $!$-E can be derived from the homset morphism from right to left:\\
\\
\small
\AxiomC{$\Delta;\Xi\vdash t:!A$}
\AxiomC{}
\RightLabel{Lin-Var}
\UnaryInfC{$\Delta;w:\Xi'\vdash w:\Xi'$}
\AxiomC{$\Delta,x:A;y:\Xi'\vdash b :B$}
\RightLabel{$\multimap$-I}
\UnaryInfC{$\Delta,x:A;\cdot \vdash \lambda_{y:\Xi'}b :\Xi'\multimap B$}
\RightLabel{"right to left"}
\UnaryInfC{$\Delta;z:!A\vdash  \mathrm{let}\; z\;\mathrm{be}\; !x\;\mathrm{in}\; \lambda_{y:\Xi'}b : \Xi'\multimap B$}
\RightLabel{$\multimap$-E}
\BinaryInfC{$\Delta;z:!A,\Xi'\vdash  \mathrm{let}\; z\;\mathrm{be}\; !x\;\mathrm{in}\; b[w/y]:B$}
\RightLabel{Lin-Tm-Subst}
\BinaryInfC{$\Delta;\Xi,\Xi'\vdash \mathrm{let}\; t\;\mathrm{be}\; !x\;\mathrm{in}\; b[w/y]: B$}
\DisplayProof
\normalsize
\quad\\
\\
Note that the $!$-C- and $!$-U-rules correspond precisely to the fact that our morphisms from left to right and from right to left define a homset isomorphism.\\
\\
Finally, it is easily verified that the condition that $\mathcal{L}(f)\circ L_\Delta\cong L_{\Delta'}\circ \mathcal{I}(f)$ corresponds exactly to the compatibility of $!$ with substitution.

\item Suppose we have $\mathrm{Id}_{!A}\dashv -\{\mathrm{diag}_{\Delta,A}\}$, i.e. we have a (natural) homset isomorphism
\[
\begin{tikzcd}[ampersand replacement=\&, column sep=large]
\mathcal{L}(\Delta.A.A\{\mathbf{p}_{\Delta,A}\})(\mathrm{Id}_{!A}(B),C) \arrow[r, shift left=0.6ex] \arrow[r, phantom,"{\cong}" description] \& \mathcal{L}(\Delta.A)(B,C\{\mathrm{diag}_{\Delta,A}\}) \arrow[l, shift left=0.6ex].
\end{tikzcd}
\]
The claim is that $\mathrm{Id}_{!A}(I)$ satisfies the rules for the $\mathrm{Id}$-type of $A$ (or, perhaps more appropriately, of $!A$). Indeed, we have $\mathrm{Id}$-I as follows.\\
\quad\\
\scriptsize
\AxiomC{}
\RightLabel{Lin-Var}
\UnaryInfC{$\Delta,x:A,x':A;w:\mathrm{Id}_{!A}(I)(x,x')\vdash w:\mathrm{Id}_{!A}(I)(x,x')$}
\RightLabel{"left to right"}
\UnaryInfC{$\Delta,x:A;y:I\vdash \mathrm{refl}_{!x}^y:\mathrm{Id}_{!A}(I)(x,x)$}
\AxiomC{}
\RightLabel{$I$-I}
\UnaryInfC{$\Delta,x:A;\cdot\vdash *:I$}
\RightLabel{Lin-Tm-Subst}
\BinaryInfC{$\Delta,x:A;\cdot\vdash \mathrm{refl}_{!x}:\mathrm{Id}_{!A}(I)(x,x)$}
\AxiomC{$\Delta;\cdot\vdash a:A$}
\LeftLabel{Int-Tm-Subst}
\BinaryInfC{$\Delta;\cdot\vdash \mathrm{refl}_{!a}:\mathrm{Id}_{!A}(I)(a,a)$}
\DisplayProof
\normalsize\quad\\
\\
We obtain $\mathrm{Id}$-E as follows. Let $\Delta,x:A,x':A;\cdot\vdash C\;\mathrm{type}$.\\
\quad\\
\tiny
\AxiomC{$\Delta,x:A;B\vdash c:C[x/x']$}
\RightLabel{"right to left"}
\UnaryInfC{$\Delta,x:A,x':A;\mathrm{Id}_{!A}(B)\vdash c':C$}
\AxiomC{$\Delta;\cdot\vdash a:A$}
\AxiomC{$\Delta;\cdot\vdash a':A$}
\RightLabel{Int-Tm-Subst}
\TrinaryInfC{$\Delta;\mathrm{Id}_{!A}(B)[a/x,a'/x']\vdash c'[a/x,a'/x']:C[a/x,a'/x']$}
\AxiomC{$\Delta;B'\vdash p:\mathrm{Id}_{!A}(I)[a/x,a'/x']$}
\LeftLabel{(*)}
\UnaryInfC{$\Delta;B[a/x],B'\vdash p':\mathrm{Id}_{!A}(B)[a/x,a'/x']$}
\LeftLabel{Lin-Tm-Subst}
\BinaryInfC{$\Delta;B[a/x],B'\vdash \mathrm{let}\; (a,a',p)\;\mathrm{be}\;(z,z,\mathrm{refl}_{!z})\;\mathrm{in}\; c:C[a/x,a'/x']$}
\DisplayProof
\normalsize
\quad\\
Here, $(*)$ is a slightly non-trivial step that follows immediately when we note that $\mathrm{Id}_{!A}(B)\cong \mathrm{Id}_{!A}(I)\otimes B\{\mathbf{p}_{\Delta,A}\}$ (by tensoring with $\mathrm{id}_B$). Indeed,
\begin{align*}\mathcal{L}(\Delta.A.A\{\mathbf{p}_{\Delta,A}\})(\mathrm{Id}_{!A}(B),C)&\cong \mathcal{L}(\Delta.A)(B,C\{\mathrm{diag}_{\Delta,A}\})\\
&\cong \mathcal{L}(\Delta.A)(I,(B\{\mathbf{p}_{\Delta,A}\}\multimap C)\{\mathrm{diag}_{\Delta,A}\})\\
&\cong \mathcal{L}(\Delta.A.A\{\mathbf{p}_{\Delta,A}\})(\mathrm{Id}_{!A}(I),B\{\mathbf{p}_{\Delta,A}\}\multimap C)\\
&\cong \mathcal{L}(\Delta.A.A\{\mathbf{p}_{\Delta,A}\})(\mathrm{Id}_{!A}(I)\otimes B\{\mathbf{p}_{\Delta,A}\}, C).
\end{align*}\normalsize
Since all these isomorphisms are natural in $C$, the Yoneda lemma says that $\mathrm{Id}_{!A}(B)\cong \mathrm{Id}_{!A}(I)\otimes B\{\mathbf{p}_{\Delta,A}\}$. \\
\\
Conversely, suppose we have $\mathrm{Id}$-types. Then, define $\mathrm{Id}_{!A}(B):=\mathrm{Id}_{!A}\otimes B\{\mathbf{p}_{\Delta,A}\}$, with the obvious extension on morphisms. Then, we obtain the morphism "left to right" as follows.\\
\quad\\
\scriptsize
\AxiomC{$\Delta,x:A,x':A;z:\mathrm{Id}_{!A},y:B\vdash c:C$}
\RightLabel{}
\AxiomC{}
\RightLabel{Int-Var}
\UnaryInfC{$\Delta,x:A;\cdot\vdash x:A$}
\RightLabel{Int-Tm-Subst}
\BinaryInfC{$\Delta,x:A;z:\mathrm{Id}_{!A}[x/x'],y:B\vdash c[x/x']:C[x/x']$}
\AxiomC{}
\RightLabel{Int-Var}
\UnaryInfC{$\Delta,x:A;\cdot\vdash x:A$}
\RightLabel{$\mathrm{Id}$-I}
\UnaryInfC{$\Delta,x:A;\cdot\vdash \mathrm{refl}_{!x}:\mathrm{Id}_{!A}(x,x)$}
\RightLabel{Lin-Tm-Subst}
\BinaryInfC{$\Delta,x:A;y:B\vdash c': C[x/x']$}
\DisplayProof
\normalsize
\quad\\
\\
The morphism "right to left" is obtained as follows.\\
\\
\tiny\hspace{-30pt}
\AxiomC{$\Delta,x_0:A; y:B\vdash c:C[x_0/x_1]$}

\AxiomC{}
\RightLabel{Lin-Var}
\UnaryInfC{$\Delta,x_0:A,x_1:A;w:\mathrm{Id}_{!A}\vdash w:\mathrm{Id}_{!A}$}
\AxiomC{}
\RightLabel{Int-Var}
\UnaryInfC{$\Delta,x_0:A,x_1:A;\cdot\vdash x_i:A$}
\RightLabel{$\mathrm{Id}$-E}
\TrinaryInfC{$\Delta,x_0:A,x_1:A; w:\mathrm{Id}_{!A},y:B\vdash c':C$}
\DisplayProof
\normalsize
\quad\\
\\
We leave it to the reader to verify that the $\mathrm{Id}$-C- and $\mathrm{Id}$-U-rules translate precisely into the "right to left" and "left to right" morphisms being inverse.
\end{enumerate}
\end{proof}

The semantics of $!$ suggests an alternative definition of comprehension: if we have $\Sigma$-types in a strong sense, then it is a derived notion.

\begin{theorem}[Lawvere Comprehension]\label{altcompr} Given a strict indexed monoidal category $(\mathcal{C},\mathcal{L})$ with left adjoints $\Sigma_{Lf}$ to $\mathcal{L}(f)$ for arbitrary $\Delta'\ra{f}\Delta\in\mathcal{C}$, we can define $\mathcal{C}/\Delta\ra{L_\Delta}\mathcal{L}(\Delta)$ by
$$L_\Delta(-):=\Sigma_{L-}I.$$
In that case, $(\mathcal{C},\mathcal{L})$ has a comprehension schema iff $L_\Delta$ has a right adjoint $M_\Delta$ (which then automatically satisfies $M_{\Delta'}\circ \mathcal{L}(f)=f^*\circ M_\Delta$ for all $\Delta'\ra{f}\Delta\in\mathcal{C}$, where $f^*$ denotes pullback along $f$ in slice categories). That is, our notion of comprehension generalises that of \cite{lawvere1970equality}.

Finally, if $\Sigma_{Lf}$ are required to satisfy the Beck-Chevalley condition and Frobenius reciprocity, then $(\mathcal{C},\mathcal{L})$ satisfies the comprehension schema iff it admits $!$-types.
\end{theorem}
\begin{proof}Suppose that we have the stated right adjoints $M_\Delta$. We will construct a comprehension schema. 

We define $\mathbf{p}_{\Delta,{A}}:=M_\Delta(A)$ and \\
\[
\begin{tikzcd}[ampersand replacement=\&, column sep=large, row sep=small]
\mathcal{L}(\Delta^{\prime})(I,A\{f\}) \arrow[r,"{\cong}"] \& \mathcal{L}(\Delta)(\Sigma_{Lf} I_{\Delta^{\prime}},A)=\mathcal{L}(\Delta)(L_\Delta f,A) \arrow[r,"{\cong}"] \& \mathcal{C}/\Delta(f,M_\Delta A) \\
a \arrow[r, maps to] \& a_f \arrow[r, maps to] \& \langle f,a\rangle,
\end{tikzcd}
\]
where the first natural isomorphism comes from the adjunction $\Sigma_{Lf}\dashv -\{f\}$ and the second one comes from the adjunction $L_\Delta\dashv M_\Delta$. Note that, by definition, $\mathbf{p}_{\Delta,{A}}\langle f,a\rangle=f$.

In particular, we obtain a unique $\mathbf{v}_{\Delta,{A}}\in \mathcal{L}(\Delta.{A})(I,A\{\mathbf{p}_{\Delta,{A}}\})$ inducing $\mathrm{id}_{M_\Delta A}$ as $\langle \mathbf{p}_{\Delta,{A}},\mathbf{v}_{\Delta,{A}}\rangle$. Finally, the Yoneda lemma (i.e. naturality of these isomorphisms) says that $\mathbf{v}_{\Delta,{A}}\{\langle f,a\rangle\}=a$.\\
\\
Conversely, suppose we are given a comprehension schema. Then, we know, by theorem \ref{thm:comprfunc}, that we can define a comprehension functor $M_\Delta$ such that $M_{\Delta'}\circ \mathcal{L}(f)=f^*\circ M_\Delta$. Then we have the following:
\[
\begin{tikzcd}[ampersand replacement=\&, column sep=large, row sep=small]
\mathcal{C}/\Delta(f,M_\Delta A) \arrow[r,"{\cong}"] \& \mathcal{L}(\Delta^{\prime})(I,A\{f\}) \arrow[r,"{\cong}"] \& \mathcal{L}(\Delta)(\Sigma_{Lf} I_{\Delta^{\prime}},A)=\mathcal{L}(\Delta)(L_\Delta f,A) \\
\langle f,a\rangle \arrow[r, maps to] \& a \arrow[r, maps to] \& a_f,
\end{tikzcd}
\]
where the first isomorphism is precisely the representation defined by our comprehension and the second isomorphism comes from the fact that $\Sigma_{Lf}\dashv -\{f\}$.

Finally, the following calculation shows that it follows from Frobenius reciprocity and Beck-Chevalley that $L_\Delta$ is strong monoidal:
\begin{align*}
L_\Delta(f)\otimes L_\Delta(g)&= (\Sigma_{Lf}I_{\mathrm{dom}f})\otimes (\Sigma_{Lg}I_{\mathrm{dom}g})\\
&= \Sigma_{Lg}((\Sigma_{Lf}I_{\mathrm{dom}f})\{g\}\otimes I_{\mathrm{dom}g})\txt{(Frobenius reciprocity)}\\
&=\Sigma_{Lg}((\Sigma_{Lf}I_{\mathrm{dom}f})\{g\})\\
&=\Sigma_{Lg}\Sigma_{L f^* g}I_{\mathrm{dom}f}\{g^*f\}\txt{(Beck-Chevalley)}\\
&=\Sigma_{Lg}\Sigma_{Lf^*g}I_{\mathrm{dom}_{f\times g}}\\
&=\Sigma_{L(f\times g)}I_{\mathrm{dom}_{f\times g}}\\
&=L_\Delta(f\times g).
\end{align*}
\end{proof}

\begin{theorem}[Type Formers in $\mathcal{I}$]\label{thm:inttyp} $\mathcal{I}$ supports $\Sigma$-types iff $\mathrm{ob}(\mathcal{I})$ is closed under compositions (as morphisms in $\mathcal{C}$). It supports $\mathrm{Id}$-types iff $\mathrm{ob}(\mathcal{I})$ is closed under post-composition with maps $\mathrm{diag}_{\Delta,A}$. If $\mathcal{L}$ supports $!$- and $\Pi$-types, then $\mathcal{I}$ supports $\Pi$-types. Moreover,
$$\Sigma_{!A}! B\cong L(\Sigma_{M A}MB) \qquad\qquad \mathrm{Id}_{!A}(!B)\cong L\mathrm{Id}_{MA}(MB)\qquad\qquad M\Pi_{!B}C\cong \Pi_{MB}MC.
$$\end{theorem}
\begin{proof}We write out the adjointness condition
\begin{align*}\mathcal{I}(\Delta)(\Sigma_{\mathbf{p}_{\Delta,B}}f,\mathbf{p}_{\Delta,D})&\stackrel{!}{\cong } \mathcal{I}(\Delta.B)(f,\mathbf{p}_{\Delta,D}\{\mathbf{p}_{\Delta,B}\})\\
&\cong \mathcal{I}(\Delta.B)(f,\mathbf{p}_{\Delta,D\{\mathbf{p}_{\Delta,B}\}})\\
&\cong \mathcal{L}(\Delta.B.C)(I,D\{\mathbf{p}_{\Delta,B}\}\{f\})\\
&\cong \mathcal{L}(\Delta.B.C)(I,D\{\mathbf{p}_{\Delta,B}\circ f\})\\
&\cong \mathcal{I}(\Delta)(\mathbf{p}_{\Delta,B}\circ f,\mathbf{p}_{\Delta,D}).
\end{align*}
Now, the Yoneda lemma gives $\Sigma_{\mathbf{p}_{\Delta,B}}f=\mathbf{p}_{\Delta,B}\circ f$.\\\\
Similarly,
\begin{align*}
\mathcal{I}(\Delta.A.A)(\mathrm{Id}_{\mathbf{p}_{\Delta,A}}(f),\mathbf{p}_{\Delta.A.A,C})&\stackrel{!}{\cong} \mathcal{I}(\Delta.A)(f,\mathbf{p}_{\Delta.A.A,C}\{\mathrm{diag}_{\Delta,A}\})\\
&\cong \mathcal{L}(\Delta.A.B)(I,C\{\mathrm{diag}_{\Delta,A}\}\{f\})\\
&\cong \mathcal{L}(\Delta.A.B)(I,C\{\mathrm{diag}_{\Delta,A}\circ f\})\\
&\cong \mathcal{I}(\Delta.A.A)(\mathrm{diag}_{\Delta,A}\circ f, \mathbf{p}_{\Delta.A.A,C}),
\end{align*}
so $\mathrm{diag}_{\Delta,A}\circ f$ models $\mathrm{Id}_{\mathbf{p}_{\Delta,A}}(f)$.\\
\\
Finally,
\begin{align*}
\mathcal{I}(\Delta)(M_\Delta D,\Pi_{\mathbf{p}_{\Delta,B}}\mathbf{p}_{\Delta.B,C})&\stackrel{!}{\cong} \mathcal{I}(\Delta.B)((M_\Delta D)\{\mathbf{p}_{\Delta,B}\},\mathbf{p}_{\Delta.B,C})\\
&\cong \mathcal{I}(\Delta.B)((M_\Delta D)\{\mathbf{p}_{\Delta,B}\},M_{\Delta.B}C)\\
&\cong \mathcal{L}(\Delta.B)(L_{\Delta.B}((M_\Delta D)\{\mathbf{p}_{\Delta,B}\}),C)\\
&\cong \mathcal{L}(\Delta.B)((L_{\Delta}M_\Delta D)\{\mathbf{p}_{\Delta,B}\},C)\\
&\cong \mathcal{L}(\Delta)(L_{\Delta}M_\Delta D,\Pi_{!B}C)\\
&\cong \mathcal{I}(\Delta)(M_\Delta D,M_{\Delta}\Pi_{!B}C)\\
\end{align*}
Again, using the Yoneda lemma, we conclude that $M_{\Delta}\Pi_{!B}C$ models $\Pi_{M_{\Delta}B}M_{\Delta.B}C$.\\
\\
In all cases, we have not addressed Beck-Chevalley (and Frobenius reciprocity for $\Sigma$-types), because these are straightforward to verify.\\
\\
Note that if $\mathcal{L}$ has $!$- and $\Sigma$-types, then
\begin{align*}\mathcal{L}(\Delta)(L_\Delta(\Sigma_{M_\Delta A}M_{\Delta.A}B),C)&\cong\mathcal{I}(\Delta)(\Sigma_{M_\Delta A}M_{\Delta.A} B,M_\Delta C)\\
&\cong \mathcal{I}(\Delta.A)(M_{\Delta.A} B,(M_\Delta C)\{\mathbf{p}_{\Delta,A}\})\\
&\cong \mathcal{I}(\Delta.A)(M_{\Delta.A} B,M_{\Delta.A}(C\{\mathbf{p}_{\Delta,A}\}))\\
&\cong \mathcal{L}(\Delta.A)(! B,C\{\mathbf{p}_{\Delta,A}\})\\
&\cong \mathcal{L}(\Delta)(\Sigma_{!A}! B,C).
\end{align*}
By the Yoneda lemma, we conclude that $\Sigma_{!A}! B\cong L_\Delta (\Sigma_{M_\Delta A}M_{\Delta.A}B)$.\\
\\
Note that, if $\mathcal{L}$ admits $!$- and $\mathrm{Id}$-types,
\begin{align*}\mathcal{L}(\Delta.A.A)(\mathrm{Id}_{!A}(!B),C)&\cong \mathcal{L}(\Delta.A)(!B,C\{\mathrm{diag}_{\Delta,A}\})\\
&\cong \mathcal{I}(\Delta.A)(M_{\Delta.A} B,M_{\Delta.A}(C\{\mathrm{diag}_{\Delta,A}\}))\\
&\cong \mathcal{I}(\Delta.A)(M_{\Delta.A} B,M_{\Delta.A.A}(C)\{\mathrm{diag}_{\Delta,A}\})\\
&\cong \mathcal{I}(\Delta.A.A)(\mathrm{diag}_{\Delta,A}\circ M_{\Delta.A} B,M_{\Delta.A.A}(C))\\
&\cong \mathcal{I}(\Delta.A.A)(\mathrm{Id}_{M_\Delta A}(M_{\Delta.A} B),M_{\Delta.A.A}(C))\\
&\cong \mathcal{L}(\Delta.A.A)(L_{\Delta.A.A}\mathrm{Id}_{M_{\Delta}A}(M_{\Delta.A}B),C).
\end{align*}
We conclude that $\mathrm{Id}_{!A}(!B)\cong L_{\Delta.A.A}\mathrm{Id}_{M_\Delta A}(M_{\Delta.A}B)$ and, in particular, $\mathrm{Id}_{!A}(I)\cong L_{\Delta.A.A}\mathrm{Id}_{M_{\Delta}A}(\mathrm{id}_{\Delta.A})$. The last statement is also easily seen to be valid in the absence of $\top$-types.
\end{proof}
\begin{remark}[Dependent Seely Isomorphisms?] Note that, in our setup, we have a version of the normal Seely isomorphisms in each fibre. Indeed, suppose $\mathcal{L}$ supports $\top$-, $\&$, and $!$-types. Then, $M_\Delta(\top)=\mathrm{id}_\Delta$ and $M_\Delta(A\& B)=M_\Delta(A)\times M_\Delta(B)$, as $M_\Delta$ has a left adjoint and therefore preserves products. Now, $L_\Delta$ is strong monoidal and $!_\Delta=L_\Delta M_\Delta$, so it follows that $!_\Delta\top=I$ and $!_\Delta(A\& B)=!_\Delta A\otimes !_\Delta B$.

Now, theorem \ref{thm:inttyp} suggests the possibility of similar Seely isomorphisms for $\Sigma$-types and $\mathrm{Id}$-types. Indeed, $\mathcal{I}$ supports $\Sigma$-types iff we have additive $\Sigma$-types in $\mathcal{L}$ in the sense of objects $\Sigma_A^\& B$ such that
$$M\Sigma_A^\& B\cong \Sigma_{MA}MB\txt{and hence} !\Sigma_A^\& B\cong \Sigma_{!A}^\otimes !B,$$
where we suggestively write $\Sigma^\otimes$ for the usual multiplicative $\Sigma$-type in $\mathcal{L}$. In an ideal world, one would hope that $\Sigma_A^\& B$ generalises $A\& B$ in the same way that $\Sigma_{!A} B$ is a dependent generalisation of $!A\otimes B$.

Similarly, we get a notion of additive $\mathrm{Id}$-types: $\mathcal{I}$ supports $\mathrm{Id}$-types iff we have objects $\mathrm{Id}_A^\&(B)$ in $\mathcal{L}$ such that
$$M\mathrm{Id}_A^\&(B)\cong \mathrm{Id}_{MA}(MB)\txt{and hence} !\mathrm{Id}_A^\&(B)\cong \mathrm{Id}_{!A}^\otimes(!B),$$
writing $\mathrm{Id}^\otimes$ for the usual (multiplicative) $\mathrm{Id}$-type in $\mathcal{L}$. Note that this suggests that, in the same way that $\mathrm{Id}_{!A}^\otimes(B)\cong \mathrm{Id}_{!A}^\otimes(I)\otimes B$ (a sense in which usual $\mathrm{Id}$-types are multiplicative connectives), $\mathrm{Id}_A^\&(B)\cong \mathrm{Id}_A^\&(\top)\& B$. In fact, if we have $\top$- and $\&$-types, we only have to give $\mathrm{Id}_A^\&(\top)$ and can then \emph{define} $\mathrm{Id}_A^\&(B):=\mathrm{Id}_A^\&(\top)\& B$ to obtain additive $\mathrm{Id}$-types in full generality.

A fortiori, if some $M_\Delta$ is essentially surjective, we obtain such additive $\Sigma$- and $\mathrm{Id}$-types in the fibre over $\Delta$. In particular, we are in this situation if $L_\cdot \dashv M_\cdot$ is the usual co-Kleisli adjunction of $!_\cdot$, where $\mathcal{I}(\cdot)=\mathcal{C}$. This shows that if we hope to obtain a model of ILDTT indexed over the co-Kleisli category, in the natural way, we need to support these additive connectives.

It remains to be seen whether these ``additive connectives'' can be understood from a syntactic point of view. Moreover, it seems that the natural models of ILDTT do not support them. Finally, it is difficult to find an intuitive interpretation of such connectives as resources. Further investigation is necessary here.\end{remark}
\clearpage
\section{Some Discrete Models: Monoidal Families}
\label{sec:dismod}
We discuss a simple class of models in terms of families with values in a symmetric monoidal category. On a logical level, the construction amounts to starting with a model $\mathcal{V}$ of a linear propositional logic and taking the cofree linear predicate logic on $\mathrm{Set}$ with values in this propositional logic. This important example illustrates how $\Sigma$- and $\Pi$-types can represent infinitary additive disjunctions and conjunctions. The model is discrete in nature, however, and in that respect not representative of the type theory.

Suppose $\mathcal{V}$ is a symmetric monoidal category. We can then consider a strict $\mathrm{Set}$-indexed category, defined through the following enriched Yoneda embedding $\mathrm{Fam}(\mathcal{V}):={\mathcal{V}}^{-}:=\mathrm{SMCat}(-,\mathcal{V})$:
\[
\begin{tikzcd}[ampersand replacement=\&, column sep=large]
\mathrm{Set}^{op} \arrow[r,"{\mathrm{Fam}(\mathcal{V})}"] \& \mathrm{SMCat} \&\& S\ra{f}S^{\prime} \arrow[r, maps to] \& \mathcal{V}^S\stackrel{-\circ f}{\longleftarrow} \mathcal{V}^{S^{\prime}}.
\end{tikzcd}
\]
Note that this definition naturally extends to a functor $\mathrm{Fam}$.
\begin{theorem}[Families Model ILDTT] The construction $\mathrm{Fam}$ adds type dependency on $\mathrm{Set}$ cofreely in the sense that it is right adjoint to the forgetful functor $\mathrm{ev}_1$ that evaluates a model of linear dependent type theory at the empty context to obtain a model of linear propositional type theory (where $\mathrm{SMCat}_{\mathrm{compr}}^{\mathrm{Set}^{op}}$ is the full subcategory of $\mathrm{SMCat}^{\mathrm{Set}^{op}}$ on the objects with comprehension):
\[
\begin{tikzcd}[ampersand replacement=\&, column sep=huge]
\mathrm{SMCat} \arrow[r, hook, shift right=1.5ex,"{\mathrm{Fam}}"'] \arrow[r, phantom,"{\bot}" description] \& \mathrm{SMCat}_{\mathrm{compr}}^{\mathrm{Set}^{op}} \arrow[l, shift right=1.5ex,"{\mathrm{ev}_1}"'] .
\end{tikzcd}
\]
\end{theorem}
\begin{proof} $\mathrm{Fam}(\mathcal{V})$ admits a comprehension by the following isomorphism
\begin{align*}
\mathrm{Fam}(\mathcal{V})(S)(I,B\{f\})&=\mathcal{V}^S(I,B\circ f)\\
&=\Pi_{s\in S}\mathcal{V}(I,B(f(s)))\\
&\cong \mathrm{Set}/S(S\ra{\mathrm{id}_S}S, \Sigma_{s\in S}\mathcal{V}(I,B(f(s)))\ra{\mathrm{fst}}S)\\
&\cong \mathrm{Set}/S'(S\ra{f}S',\Sigma_{s'\in S'}\mathcal{V}(I,B(s'))\ra{\mathrm{fst}}S')\\
&=\mathrm{Set}/S'(f,\mathbf{p}_{S',{B}}),
\end{align*}
where $\mathbf{p}_{S',{B}}:=\Sigma_{s'\in S'}\mathcal{V}(I,B(s'))\ra{\mathrm{fst}}S'$. ($\mathbf{v}_{S',{B}}$ is the element corresponding to $\mathrm{id}_{\mathbf{p}_{S',{B}}}\in \mathrm{Set}/S'(\mathbf{p}_{S',{B}},\mathbf{p}_{S',{B}})$ under this isomorphism.)
To see that $\mathrm{ev}_1\dashv \mathrm{Fam}$, note that the following naturality diagrams for elements $1\ra{s}S$
\[
\begin{tikzcd}[ampersand replacement=\&, column sep=large, row sep=large]
1 \arrow[d,"{s}"'] \& \mathrm{ev}_1(\mathcal{L})=\mathcal{L}(1) \arrow[r,"{\phi_1}"] \& \mathcal{V}=\mathrm{Fam}(\mathcal{V})(1)\\
S \& \mathcal{L}(S) \arrow[u,"{-\{s\}}"] \arrow[r,"{\phi_S}"'] \& \mathcal{V}^S=\mathrm{Fam}(\mathcal{V})(S) \arrow[u,"{-\circ s}"']
\end{tikzcd}
\]
together with the fact that all $1\ra{s}S$ are jointly surjective (and hence that the functors $-\circ s$ are jointly injective) mean that a natural transformation $\phi\in\mathrm{SMCat}^{\mathrm{Set}^{op}}(\mathcal{L},\mathrm{Fam}(\mathcal{V}))$ is uniquely determined by $\phi_1\in\mathrm{SMCat}(\mathrm{ev}_1(\mathcal{L}),\mathcal{V})$.
\end{proof}

We have the following results for type formers\footnote{We do not examine $\mathrm{Id}$-types here, as they correspond precisely to the intuitionistic identity type in $\mathcal{I}$, which is probably of limited interest, since $\mathcal{I}$ is a submodel of the normal set-based model of dependent types (i.e. fibred sets, which is equivalent to indexed sets: $\mathrm{Set}$-valued families).}.

\begin{theorem}[Type Formers for Families] $\mathcal{V}$ has small coproducts that distribute over $\otimes$ iff $\mathrm{Fam}(\mathcal{V})$ supports $\Sigma$-types. In that case, $\mathrm{Fam}(\mathcal{V})$ also supports $0$- and $\oplus$-types (which correspond precisely to finite distributive coproducts).

$\mathcal{V}$ has small products iff $\mathrm{Fam}(\mathcal{V})$ supports $\Pi$-types. In that case, $\mathrm{Fam}(\mathcal{V})$ also supports $\top$- and $\&$-types (which correspond precisely to finite products).

$\mathrm{Fam}(\mathcal{V})$ supports $\multimap$-types iff $\mathcal{V}$ is monoidal closed. 

$\mathrm{Fam}(\mathcal{V})$ supports $!$-types iff $\mathcal{V}$ has small coproducts of $I$ that are preserved by $\otimes$ in the sense that the canonical morphism $\mathrm{coprod}_S(\Xi'\otimes I)\ra{}\Xi'\otimes \mathrm{coprod}_S I$ is an isomorphism for any $\Xi'\in\mathrm{ob}\;\mathcal{V}$ and $S\in\mathrm{ob}\;\mathrm{Set}$. In particular, if $\mathrm{Fam}(\mathcal{V})$ supports $\Sigma$-types, then it also supports $!$-types.

$\mathrm{Fam}(\mathcal{V})$ supports $\mathrm{Id}$-types if $\mathcal{V}$ has an initial object. If $\mathcal{V}$ has a terminal object, the only-if direction also holds.
\end{theorem}\begin{proof}The statement about $0$-, $\oplus$-, $\top$-, and $\&$-types should be clear from the previous sections, as products and coproducts in $\mathcal{V}^S$ are pointwise (and hence automatically preserved under substitution).

We will denote coproducts in $\mathcal{V}$ by $\bigoplus$. Then,
\begin{align*}
\Pi_{s'\in S'}\mathcal{V}(\bigoplus_{s\in f^{-1}(s')}A(s),B(s'))
&\cong \Pi_{s'\in S'}\Pi_{s\in f^{-1}(s')}\mathcal{V}(A(s),B(s'))\\
&\cong \Pi_{s\in \Sigma_{s'\in S'}f^{-1}(s')}\mathcal{V}(A(s),B(f(s)))\\
&\cong \Pi_{s\in S}\mathcal{V}(A(s),B(f(s)))\\
&=\mathcal{V}^S(A,B\circ f).
\end{align*}
Thus, if we have coproducts, we can define $\Sigma_{Lf}(A)(s'):=\bigoplus_{s\in f^{-1}(s')}A(s)$ to get a left adjoint $\Sigma_{Lf}\dashv -\{f\}$. The definition on morphisms is the obvious one coming from the coCartesian monoidal structure on $\mathcal{V}$. Conversely, we can use $\Sigma_{Lf}$ to define any coproduct by taking $f$ to be the unique function from the indexing set to $1$ and taking $A$ to be the family of objects whose coproduct we want. The Beck-Chevalley condition is handled by the fact that our substitution morphisms are given by precomposition. Frobenius reciprocity precisely corresponds to distributivity of the coproducts over $\otimes$.

Similarly, if $\mathcal{V}$ has products, we will denote them by $\bigwith$ to suggest the connection with linear type theory. In that case, we can define $\Pi_{Lf}(A)(s'):=\bigwith_{s\in f^{-1}(s')}A(s)$ to get a right adjoint $-\{f\}\dashv \Pi_{Lf}$. The definition on morphisms is the obvious one coming from the Cartesian monoidal structure on $\mathcal{V}$. Indeed,
\begin{align*}
\Pi_{s'\in S'}\mathcal{V}(B(s'),\bigwith_{s\in f^{-1}(s')}A(s))&\cong \Pi_{s'\in S'}\Pi_{s\in f^{-1}(s')}\mathcal{V}(B(s'),A(s))\\
&\cong \Pi_{s\in \Sigma_{s'\in S'}f^{-1}(s')}\mathcal{V}(B(f(s)),A(s))\\
&\cong \Pi_{s\in S}\mathcal{V}(B(f(s)),A(s))\\
&=\mathcal{V}^S(B\circ f,A).
\end{align*}
Again, in the same way as before, we can construct any product using $\Pi_{Lf}$ along the unique function from the indexing set to $1$. The dual Beck-Chevalley condition is automatic because our substitution morphisms are precomposition.

The claim about $\multimap$-types follows immediately from the previous section: $\mathrm{Fam}(\mathcal{V})$ supports $\multimap$-types iff all its fibres have a monoidal closed structure that is preserved by the substitution functors. Since our monoidal structure is pointwise, the same holds for any monoidal closed structure. Since substitution is given by precomposition, the preservation requirement is automatic.

The characterisation of $!$-types is given by theorem \ref{thm:!fromsigma}, which tells us we can define $!A:=\Sigma_{\mathbf{p}_{S',{A}}}I=s'\mapsto\bigoplus_{\mathcal{V}(I,A(s'))}I$ and conversely.

Finally, for $\mathrm{Id}$-types, note that the adjointness condition $\mathrm{Id}_{!A}\dashv -\{\mathrm{diag}_{\Delta,A}\}$ amounts to the requirement (*)
\begin{align*}\Pi_{s\in S}\Pi_{a\in A(s)}\mathcal{V}(B(s,a),C(s,a,a))&\cong\mathcal{V}^{\Sigma_{s\in S}A(s)}(B,C\{\mathrm{diag}_{S,A}\})\\
&\stackrel{!}{\cong} \mathcal{V}^{\Sigma_{s\in S}A(s)\times A(s)}(\mathrm{Id}_{!A}(B),C)\\
&\cong \Pi_{s\in S}\Pi_{a\in A(s)}\Pi_{a'\in A(s)}\mathcal{V}(\mathrm{Id}_{!A}(B)(s,a,a'),C(s,a,a')).
\end{align*}
We see that if we have an initial object $0\in\mathrm{ob}(\mathcal{V})$, we can define $$\mathrm{Id}_{!A}(B)(s,a,a'):=\left\{\begin{array}{l}B(s,a)\textnormal{ if $a=a'$}\\ 0\textnormal{ else} \end{array}\right.$$
For a partial converse, suppose we have a terminal object $\top\in \mathcal{V}$. Let $V\in\mathrm{ob}(\mathcal{V})$. Let $S:=\{*\}$, $A:=\{0,1\}$, take $B$ constantly $\top$, and choose $C$ such that $C(0,0)=C(1,1)=C(0,1)=\top$ and $C(1,0)=V$. Then, (*) becomes the condition that $\{*\}\cong \mathcal{V}(\mathrm{Id}_{!A}(B)(1,0),V)$. We conclude that $\mathrm{Id}_{!A}(B)(1,0)$ is initial in $\mathcal{V}$.
\end{proof}

\begin{remark}Note that an obvious way to guarantee distributivity of coproducts over $\otimes$ is by requiring $\mathcal{V}$ to be monoidal closed.
\end{remark}
\begin{remark}It is easily seen that $\Sigma$-types in $\mathcal{I}$, or additive $\Sigma$-types in $\mathcal{L}=\mathrm{Fam}(\mathcal{V})$, amount to having an object $\mathrm{or}_{s\in S}C(s)\in\mathrm{ob}(\mathcal{V})$ for a family $(C(s)\in \mathrm{ob}(\mathcal{V}))_{s\in S}$ such that $\Sigma_{s\in S}\mathcal{V}(I,C(s))\cong \mathcal{V}(I,\mathrm{or}_{s\in S} C(s))$. Similarly, $\mathrm{Id}$-types in $\mathcal{I}$, or additive $\mathrm{Id}$-types in $\mathcal{L}$, amount to having objects $\mathrm{one},\mathrm{zero}\in\mathrm{ob}(\mathcal{V})$ such that $\mathcal{V}(I,\mathrm{one})\cong 1$ and $\mathcal{V}(I,\mathrm{zero})=0$.
\end{remark}
Two particularly simple concrete choices of $\mathcal{V}$ can accommodate all type formers and serve as useful illustrations: a category $\mathcal{V}=\mathrm{Vect}_F$ of vector spaces over a field $F$, with the tensor product, and the category $\mathcal{V}=\mathrm{Set}_*$ of pointed sets, with the smash product. All type formers have their obvious interpretation, but let us pause to consider $!$, since a novelty of ILDTT is that it is uniquely determined by the indexing, while in propositional linear type theory we might have several different choices. In the first example, $!$ amounts to the following: $(!B)(s')=\mathrm{coprod}_{ \mathrm{Vect}_F(F,B(s'))}F\cong   \bigoplus_{ B(s')}F$, i.e. the vector space freely spanned by all the vectors. In the second example, $(!B)(s')=\mathrm{coprod}_{\mathrm{Set}_*(2_*,B(s'))}2_*=\bigvee_{B(s')}2_*=B(s')+\{*\}$, i.e. $!$ freely adds a new basepoint. These models show the following.
\begin{theorem}[DTT, DILL$\subsetneq$ILDTT] ILDTT is a proper generalisation of DTT and DILL: we have inclusions of the classes of models DTT, DILL$\subsetneq$ILDTT.\end{theorem}
\begin{proof}By the Grothendieck construction, every split fibration can be seen equivalently as the category of elements of the corresponding strict indexed category defined by the fibres. Under this equivalence, split full comprehension categories with finite fibrewise products (i.e. models of DTT with $1$- and $\times$-types) correspond precisely to strict indexed Cartesian monoidal categories with comprehension whose comprehension functor is full and faithful. These are a special case of our notion of model of ILDTT. Moreover, in such cases, $!A\cong A$. From their categorical description, the other connectives of ILDTT reduce to those of DTT. This proves the inclusion DTT$\subset$ILDTT.

The models described above are more general than those of DTT, as we are dealing with a non-Cartesian monoidal structure on the fibre categories. This proves that the inclusion is proper.

We have seen that the $\mathrm{Fam}$-construction realises the category of models of DILL as a reflective subcategory of the category of models of ILDTT. Moreover, the existence of various non-trivial models of DTT indexed over categories other than $\mathrm{Set}$ shows that this inclusion is proper as well.

Finally, we note that these inclusions remain valid in the sub-algebraic setting where we do not have $I$- and $\otimes$-types. A simple variation of the argument using multicategories rather than monoidal categories suffices.\end{proof}
Although this class of family models is important, it captures only a very limited part of the generality of ILDTT: not every model of ILDTT is a model of either DTT or DILL. Hence, we need models that are less discrete in nature but still linear, if we hope to observe interesting new phenomena arising from the connectives of linear dependent type theory. Some suggestions and work in progress will be discussed in the next section.
\clearpage

\section{Conclusions and Future Work}
We hope to have convinced the reader that linear dependent types fit naturally into the landscape of existing type theories and that they admit a rich theory rather than being limited to the specific examples that had been considered so far. There is a larger story connecting these examples.

On a syntactic level, our system is a natural blend of (intuitionistic) dependent type theory and dual intuitionistic linear logic. On a semantic level, if one starts with the right notion of model for dependent types, the linear generalisation is obtained through the usual philosophy of passing from Cartesian to symmetric monoidal structures. The resulting notion of a model forms a natural blend of comprehension categories, modelling DTT, and linear-non-linear models, modelling DILL.

It is pleasing to see that all the syntactically natural rules for type formers are equivalent to the semantic counterparts expected from the traditions of categorical logic for dependent and linear types. In particular, from the point of view of logic, it is interesting to see that the categorical semantics seems to have a preference for multiplicative quantifiers.

Finally, we have shown that, as in the intuitionistic case, we can represent infinitary (additive) disjunctions and conjunctions in linear type theory, through cofree $\Sigma$- and $\Pi$-types, indexed over $\mathrm{Set}$. In particular, this construction exhibits a family of non-trivial, truly linear models of dependent types, providing an essential reality check for our system.\\
\\
Despite what this paper might suggest, much of this work has been motivated by semantics, and specifically by models. In joint work with Samson Abramsky, a model of linear dependent types with comprehension has been constructed in a category of coherence spaces. Apart from the usual type constructors from linear logic, it also supports $\Sigma$-, $\Pi$-, and $\mathrm{Id}$-types. A detailed account of this model will be made available soon.

In addition to providing what is, as far as we are aware, the first non-trivial, semantically motivated model of such a type system, this work serves as a stepping stone for a model that we are currently developing in a category of games, together with Samson Abramsky and Radha Jagadeesan. In particular, this should provide a game semantics for dependent type theory. 

An indexed category of spectra up to homotopy over topological spaces has been studied, e.g., in \cite{may2006parametrized,ponto2012duality} as a setting for stable homotopy theory. It has been shown to admit $I$-, $\otimes$-, $\multimap$-, and $\Sigma$-types. The natural candidate for a comprehension adjunction here is the one between the infinite suspension spectrum and the infinite loop space: $L \dashv M \;\; = \;\; \Sigma^\infty\dashv \Omega^\infty$. A detailed examination of the situation and an explanation of the relation with the Goodwillie calculus would be desirable. This might fit with our related objective of giving a linear analysis of homotopy type theory.

Another fascinating possibility is that of models related to quantum mechanics. Non-dependent linear type theory has found interesting interpretations in quantum computation, e.g. \cite{AbrDun:CQLv2:2004}. The question arises whether the extension to dependent linear types has a natural counterpart in physics. In \cite{schreiber2014quantization}, Urs Schreiber has recently sketched how linear dependent types can serve as a language for discussing quantum field theory and quantisation in particular. There are many interesting open questions here.

Finally, many theoretical questions remain within the type theory. Can we expect interesting models with type dependency on the co-Kleisli category of $!$, and can we make sense of additive $\Sigma$- and $\mathrm{Id}$-types, e.g. from a syntactic point of view? Is there an equivalent of strong/dependent E-rules for ILDTT? Does the Curry-Howard correspondence extend fully: do we have a propositions-as-types interpretation of linear predicate logic in ILDTT? These questions need to be addressed by a combination of research into the formal system and the study of specific models. We hope that the general framework we sketched will play a part in connecting all the different aspects of the story: from syntax to semantics; from computer science and logic to geometry and physics.
\clearpage

\bibliographystyle{splncs}
\bibliography{tau}

\end{document}